\documentclass[12pt]{article}
\usepackage{amssymb}
\usepackage{amsfonts}
\usepackage{amsmath}
\usepackage{geometry}
\usepackage[onehalfspacing]{setspace}

\newtheorem{theorem}{Theorem}

\newtheorem{corollary}{Corollary}

\newtheorem{definition}{Definition}

\newtheorem{lemma}{Lemma}

\newtheorem{proposition}{Proposition}

\newenvironment{proof}[1][Proof]{\noindent\textbf{#1.} }{\ \rule{0.5em}{0.5em}}
\input{tcilatex}
\let\OLDthebibliography\thebibliography
\renewcommand\thebibliography[1]{
  \OLDthebibliography{#1}
  \setlength{\parskip}{0pt}
  \setlength{\itemsep}{0pt plus 0.3ex}
}

\begin{document}

\title{\textbf{Robust Equilibria in General Competing Mechanism Games\thanks{%
First version: April 13, 2021}}}
\author{Seungjin Han\thanks{%
This version subsumes and expands earlier versions. It also corrects errors
in the discussion in earlier versions on the equilibrium characterization in
standard competing mechanism games. I thank Siyang Xiong for his comments.
Alex Sam provided excellent research assistance. I thankfully acknowledge
financial support from the Social Sciences and Humanities Research Council
of Canada. Address: Department of Economics, McMaster University, 1280 Main
Street West, Hamilton, ON, CANADA. Email: hansj@mcmaster.ca.} \\
McMaster University}
\date{August 14, 2023}
\maketitle

\begin{abstract}
This paper proposes the notion of robust PBE in a general competing
mechanism game of incomplete information where a mechanism allows its
designer to send a message to himself at the same time agents send messages.
It identifies the utility environments where the notion of robust PBE
coincides with that of strongly robust PBE (Epstein and Peters (1999), Han
(2007)) and with that of robust PBE respectively. If each agent's utility
function is additively separable with respect to principals' actions, it is
possible to provide the full characterization of equilibrium allocations
under the notion of robust PBE and its variations, in terms of Bayesian
incentive compatible (BIC) direct mechanisms, without reference to the set
of arbitrary general mechanisms allowed in the game. However, in the
standard competing mechanism agme, the adoption of robust PBE as the
solution concept does not lead to the full characterization of equilibrium
allocations in terms of BIC direct mechanisms even with agents' separable
utility functions.
\end{abstract}

\section{Introduction}

The Revelation Principle does not hold when multiple principals compete in
the market (Epstein and Peters (1999), McAfee (1992)). When a principal
(e.g., a seller) designs a mechanism, he is aware that agents (e.g., buyers)
have not only private information on their payoff types but also market
information on mechanisms offered by competing principals, their beliefs on
the terms of trade determined in those mechanisms, etc. Therefore, it is
generally not enough to ask agents to report their payoff types only.

The literature on competing mechanisms or contractible contracts has
explored the big idea that principals can use mechanisms or contractible
contracts that involve commitments to punish a competing principal's
deviation. Punishment can take place when agents' messages reveal a
competing principal's deviation (Epstein and Peters (1999), Yamashita
(2010), Xiong (2013)), or it is prescribed in a contractible contract,
directly conditional on a deviator's publicly observable contract (Szentes
(2015)). This makes it possible for principals to sustain various collusive
outcomes.

Epstein and Peters (1999) propose a class of universal mechanisms that allow
agents to describe competing principals' mechanisms. However, the message
spaces in these universal mechanisms that agents use in describing competing
principals' mechanisms are quite complex due to the infinite regress
problem: my mechanism depends on the competitor's mechanism, his mechanism
depends on mine, and so on ad infinitum. A major contribution by Yamashita
(2010) is that any equilibrium allocations can be understood as if a
principal offers a recommendation mechanism that delegates a choice of his
incentive compatible direct mechanism to agents. The incentive compatible
direct mechanism chosen by a majority of agents is implemented.
Recommendation mechanisms can be thought of as an extension of menus
proposed for common agency (Martimort and Stole (2002), Page and Monterio
(2003), and Peters (2001)).

On the other hand, Szentes (2015) studies a bilateral contracting game as in
Han (2006), but augmented by contractible contracts (Peters and Szentes
(2012)) in which a principal's ordinary contract is directly conditional on
a deviator's ordinary contract described by the G\"{o}del language.

\subsection{Motivation}

One difficulty associated with the equilibrium analysis in the competing
mechanism games is that the characterization of the equilibrium allocations
in the competing mechanism games is tied to the very set of arbitrary
mechanisms allowed for each principal in the game. Therefore, the set of
equilibrium allocations is generally not expressed in terms of incentive
compatible direct mechanisms. Furthermore, it is not clear whether a
principal wants to deivate to a mechanism that is not allowed in the game.

Another difficulty arises from the adoption of the modeling convention from
the mechanism design in a model for a single principal. At least, it is
conceptually feasible that a principal offer a set-valued mechanism in which
agents' messages only specify a subset of actions from which the principal
subsequently chooses any action he wants. However, the literature focuses on
a standard single-valued mechanism in which agents' messages fully determine
the principal's action even though a principal, in principle, may offer a
set-valued mechanism. This is because any equilibrium allocation given a
set-valued mechanism can be reproduced by a corresponding incentive
compatible direct mechanism, thanks to the Revelation Principle. However,
Szentes (2010) points out that a principal may want to use a set-valued
mechanism in the competing mechanism games of complete information, in order
to keep the freedom of choosing his action at a later stage of the game.

This paper studies how to tackle these difficulties in the competing
mechanism games of incomplete information where a principal's general
mechanism allows him to send a message to himself at the same time agents
send messages and these messages altogether determine an action. A
principal's set-valued mechanism that assigns a subset of actions
conditional on agents' messages is strategically equivalent to a general
mechanism. Therefore, the set of set-valued mechanisms can be nested into
the set of general mechanisms.

\subsection{Main Results}

This paper identifies the utility environment and various notions of robust
equilibrium under which the set of equilibrium allocations is characterized
by incentive compatible direct mechanisms in the competing mechanism game
where a principal can offer any general mechanism.

We first propose the notion of \emph{robust} \emph{perfect Bayesian
equilibrium} (robust PBE) in which any principal $j$ cannot gain regardless
of a profile of agents' continuation equilibrium strategies of communicating
with $j$ upon his deviation to any general mechanism, taking as given a
profile of their continuation equilibrium strategies of communicating with
non-deviating principals.

The notion of robust PBE is weaker than the notion of strongly robust PBE
(Epstein and Peters (1999), Han (2007)) in which a principal cannot gain in
all continuation equilibria agents may play upon his deviation. However, it
is stronger than the notion of PBE because an allocation can be supported in
a PBE as long as there exists a continuation equilibrium in which a
principal cannot gain upon deviation. I first identify the utility
environments where the notion of robust PBE coincides with that of strongly
robust PBE and with that of PBE respectively. If each agent's utility
function is additively separable with respect to principals' actions and
each principal's utility function satisfies the \textquotedblleft no action
externalities\textquotedblright\ condition\footnote{%
The \textquotedblleft no action externalities\textquotedblright\ condition
is satisfied if each principal $j$ is indifferent between any action
profiles that the other principals may choose, regardless of his own action
and the agents' types.}, then the notion of robust PBE coincides with that
of strongly robust PBE. The no action externalities condition for each
principal's utility function is not sufficient for the equivalence between
robust PBE and strongly robust PBE. The reason is that the set of
continuation equilibrium strategies of agents' communication with principal $%
j$ may depend on continuation strategies of their communication with the
other principals. The separability condition for each agent's utility
function breaks down such dependence.

If each principal's utility function satisfies the \textquotedblleft pure
externalities\textquotedblright\ condition, then the notion of robust PBE
coincides with that of PBE. The separability condition for the agent's
utility function is not required for the equivalence between robust PBE and
PBE.\footnote{%
The pure externalities condition implies that each principal is indifferent
between any pair of his own actions, regardless of the other principals'
actions and agents' types. Even though each principal is indifferent between
any pair of his own actions, he still faces a strategic choice of his
mechanism because his mechanism affects agents' strategies of communicating
with the other principals, which eventually affects his utility through the
other principals' action choices in their mechanisms.}

One interesting variant of robust PBE is that each agent's strategy of
communicating with a non-deviating principal does not depend on the other
principals' mechanisms. That is, her communication with a non-deviating
principal is invariant to market information. If a PBE is robust and
invariant to market information, I call it a \emph{market-information
invariant robust PBE}.

I also formulate a notion of \emph{uniquely implementable robust PBE}. A
robust PBE is uniquely implementable if any profile of agents' continuation
equilibrium strategies of communicating with a non-deviating principal
induces the same BIC direct mechanism on the path or off the path following
a competing principal's deviation respectively. If a robust PBE is uniquely
implementable, the notion of robust PBE coincides with that of strongly
robust PBE. Once players reach a uniquely implementable robust PBE, not only
is there no continuation equilibrium where a principal can gain off the
equilibrium path following his deviation, but every possible continuation
equilibrium on the equilibrium path induces the same equilibrium allocation.

I show that if each agent's utility function is additively separable with
respect to principals' actions, it is possible to provides the full
characterization of equilibrium allocations under various notions of robust
PBE in terms of Bayesian incentive compatible (BIC) direct mechanisms,
without reference to the set of arbitrary mechanisms allowed in the game.

Given each agent's separable utility function, I can define an individually
Bayesian incentive compatible (BIC) direct mechanism for each principal $j$.
Then, a BIC profile of principals' direct mechanisms is a profile of
principals' individually BIC direct mechanisms.

I then consider the competing BIC direct mechanism game where a principal is
restricted to offer only a BIC direct mechanism. Not only does a
(truth-telling) PBE in the competing BIC direct mechanism game exist but it
is also robust. Furthermore, the set of market-information invariant robust
PBE allocations in any general competing mechanism game is the same as the
set of a (truth-telling) PBE allocation, denoted by $\Pi _{B}^{\ast }$, in
the competing BIC direct mechanism game. Therefore, any (truth-telling) PBE
in the competing BIC direct mechanism game is free from Szentes' critique.

For the characterization of uniquely implementable robust PBE allocations, I
extend the result from Bayesian implementation (Palfrey an Srivastava
(2018)). That is, if a profile of direct mechanisms $\pi ^{\ast }=\left( \pi
_{1}^{\ast },\ldots ,\pi _{J}^{\ast }\right) $ is uniquely implementable in
a robust PBE, then (i) $\pi ^{\ast }$ is a (truth-telling) PBE allocation in
the competing BIC direct mechanism game (i.e., $\pi ^{\ast }\in \Pi
_{B}^{\ast }$) and (ii) each $\pi _{j}^{\ast }$ satisfies Bayesian
monotonicity.\footnote{%
Condition (i) in a single-principal case is just the incentive compatibility
of the direct mechanism. See Definition 3.1 for Bayesian monotonicity in
Palfrey and Srivastava (2018).} Let $\Pi _{UB}^{\ast }$ be the set of
profiles BIC direct mechanisms that satisfies (i) and (ii). I can show $\Pi
_{UB}^{\ast }$ is indeed the set of uniquely implementable robust PBE
allocations in a general equilibrium when well-established additional
conditions on the information structure on the type space are satisfied (see
Palfrey an Srivastava (2018)).

Generally, the set of robust PBE allocations is big. Specifically, it is the
set of all profiles of BIC direct mechanisms that induces each principal $j$%
's utility no less than his \emph{minmax} utility value over BIC direct
mechanisms (Let's denote it by $\Pi _{RB}^{\ast }$). Let us explain the key
idea. In a general competing mechanism game, principal $j$ can offer a
mechanism where the message space for himself is the set that includes all
of his BIC direct mechanisms $\Pi _{j}^{B}$ and the message space for each
agent is her type space. In this mechanism, principal $j$ chooses his BIC
direct mechanism $\pi _{j}$ from $\Pi _{j}^{B}$ and agents send type
messages to $j$ at the same time and subsequently $\pi _{j}$ determines $j$%
's action conditional on agents' messages. The notion of robust PBE makes it
possible for $j$ to implement any BIC direct mechanism $\pi _{j}$ from $\Pi
_{j}^{B}$. Given $j$'s belief on BIC direct mechanisms induced by
non-deviators' mechanisms, principal $j$ can then choose his BIC direct
mechanism $\pi _{j}$ from $\Pi _{j}^{B}$ that maximizes his utility.
Therefore, principal $j$'s utility in a robust PBE of a general competing
mechanism game cannot be lower than his minmax values over individually BIC
direct mechanisms.

This paper then shows that any profile of BIC direct mechanisms that
generates a utility for each principal no less than his minmax value can be
supported in a robust PBE where each principal $j$ offers a
deviator-reporting mechanism which asks agents to report only the identity
of a deviating principal along with their payoff types. If a majority of
agents report $j$, then deviator-reporting mechanisms assign individually
BIC direct mechanisms that lead to $j$'s minmax value: Otherwise, they
assign individually BIC direct mechanisms that are supposed to be
implemented in an equilibrium.

Note that we have that 
\begin{equation*}
\Pi _{UB}^{\ast }\subset \Pi _{B}^{\ast }\subset \Pi _{RB}^{\ast }.
\end{equation*}%
All three sets are independent of the set of general mechanisms allowed in
competing mechanism games, and they are all specified in terms of BIC direct
mechanisms. One may view my paper as one that identifies the notions of PBEs
that produce those sets of allocations as the sets of equilibrium
allocations in a general competing mechanism game respectively.

As clearly pointed out in Section \ref{sec:discussion}, in the standard
competing mechanism game, the adoption of robust PBE as the solution concept
does not lead to the full characterization of equilibrium allocation in
terms of BIC direct mechanisms even with the separable utility functions for
the agents. This is because generally we do not know the set of principal $j$%
's BIC direct mechanisms that can be induced by all possible continuation
equilibrium strategies of communicating with principal $j$ given any
arbitrary standard mechanism in which agents' messages fully determines $j$%
's action.

Such a difficulty does not prevent the full characterization of robust PBE
allocations in general competing mechanism games. The reason is that
principal $j$ can alway offer a general mechanism in which the message space
for himself is the set of all possible BIC direct mechanisms $\Pi _{j}^{B}$
and the message space for each agent is her type space. Given the separable
utility functions for agents, the set of all possible BIC direct mechanisms
for $j$ induced by all continuation-equilibrium communication strategies is
the equal to $\Pi _{j}^{B}$. This makes that principal $j$'s robust PBE
utility cannot be lower than his minmax utility value over BIC direct
mechanisms. I then show that in the general competing mechanism game, any
profile of BIC direct mechanisms that provides each principal $j$ with a
utility no less than his minmax utility value over BIC direct mechanisms can
be supported in a robust PBE where each principal $j$ offers a
deviator-reporting mechanism.

Section \ref{sec:literature} provides the related literature review. Section %
\ref{sec:model} sets out the model for general competing mechanism games.
Section \ref{sec_equilibrium_notion} formalizes the three notions of robust
PBE. In Section \ref{sec:competing_BIC}, I characterize the set of
market-information invariant PBE allocations and the set of uniquely
implementable robust PBE allocations. In Section \ref{sec: robust_PBE}, I
characterize the set of robust PBE allocations. Section \ref{sec:discussion}
concludes the paper.

\subsection{Related literature\label{sec:literature}}

While this paper imposes the separability on the agent's utility function,
it allows for interdependent values of payoff types, an agent's utility
depends on actions chosen by all principals, and there are no restrictions
on the principal's utility function. In my companion paper (Han (2022)), I
considered a particular environment where (i) an agent can choose only one
principal for trading and her utility depends on only the action of the
principal she chooses and her type (i.e. private value) and (ii) utility
functions for all players are linear and there are also no externalities on
the principal side. Han (2022) focuses on the specific (but tractable)
competition model where a principal's general contract chooses his
dominant-strategy incentive compatible direct mechanism as a function of
messages sent by all agents before they choose their trading principals;
each agent subsequently sends her true type to the principal she chooses.%
\footnote{%
Extending the DIC-BIC equivalence in Gershkov, et al. (2013), Han (2022)
shows that when all the other principals offer such general contracts, a
principal has no incentive to deviate to any arbitrary general contract in
the sense that there exists a continuation equilibrium that punishes a
principal upon his deviation to an arbitrary general contract.} The focus in
Han (2022) is to characterize the equilibrium allocations conditional on
truthful type reporting to the trading principal in the specific competition
model, show that those equilibria are free from Szentes' critique, and
provide comparative statics on equilibrium profits and prices in
applications.

In common agency (i.e., multiple principals and a single agent), Han (2012),
Peters (2003, 2007), and Attar et al. (2008) identify the restrictions on
the players' preference ordering with which there is no loss of generality
for principals to focus on take-it-or-leave-it offers or
incentive-compatible direct mechanisms. The \textquotedblleft no
externalities\textquotedblright\ condition imposed in Peters (2003, 2007)
focused on games where an agent's payoff type does not affect her preference
ordering and it does not allow that a principal's utility depends on
competing principals' actions. A separability on the agent's preference
ordering is imposed in Attar et al. (2008) for common agency.

I propose similar restrictions on players' preference ordering as theirs for
the equivalence between the notion of robust PBE and that of strongly robust
PBE. However, the \textquotedblleft no action
externalities\textquotedblright\ condition in my paper does not require that
an agent's payoff type do not affect her preference ordering, which is
required in the \textquotedblleft no externalities\textquotedblright\
condition (Peters (2003, 2007)). The \textquotedblleft
separability\textquotedblright\ condition in my paper also does not impose
the no indifference condition in Attaer, et al. (2008) that an agent is
never indifferent between any pair of actions from any principals.\footnote{%
I believe that for common agency as in Attar, et al. (2008), only the
separability condition is needed for the Revelation Principle for robust
equilibrium formulated in our paper without the no indifference condition.}
In addition, the monotonicity condition in Han (2012) is not required.

In contrast to these papers, I formulate the notion of robust equilibrium
such that it is sufficient to impose only the separability on agents'
preference ordering for the characterization of the set of all robust
equilibrium allocations in terms of BIC direct mechanisms in general
competing mechanism games with multiple principals and multiple agents.

It is worthwhile to mention that mechanism design approaches considered in
Szentes (2010) and this paper are different from contracting with limited
commitment considered by Bester and Strausz (2000, 2010) and Dorval and
Skreta (2018) in a model with a single principal. This literature considers
the situations where part of a principal's action is not contractable.
Because of that, the standard Revelation Principle does not hold even with a
single principal (Bester and Stausz (2000)).

I still assume the full commitment in this paper. In fact, the set of
standard mechanisms is a strict subset of the set of general mechanisms. As
pointed out in footnote \ref{foot_RP}, when there is only one principal in a
model, the standard Revelation Principle holds for the class of mechanisms
considered in this paper.

\section{General Competing Mechanisms\label{sec:model}}

Throughout the paper, I assume that a set is a compact metric space. When a
measurable structure is necessary, the corresponding Borel $\sigma $ algebra
is used. For a set $Z$, $\Delta \left( Z\right) $ denotes the set of Borel
probability distributions on $X$. For any mapping $L$ from $S$ into $Q$, $%
L(S)$ denotes the image set of $L$.

There are $J$ principals and $I$ agents. Let $\mathcal{J}=\{1,\ldots ,J\}$
be the collection of principals and $\mathcal{I}=\{1,\ldots ,I\}$ the the
collection of agents with $J\geq 2$ and $I\geq 3.$ Agent $i$'s payoff type
is denoted by $x_{i}\in X_{i}.$ Let $x=(x_{1},\ldots ,x_{I})\in
X=X_{1}\times \cdots \times X_{I}$ denote a profile of types, which is drawn
from a probability distribution $F\in \Delta (X).$ Let $A_{j}^{\circ }$ be
the set of underlying actions that principal $j$ can take and $A_{j}:=\Delta
(A_{j}^{\circ }).$ Let $a=(a_{1},\ldots ,a_{I})\in A=A_{1}\times \cdots
\times A_{J}$ denote a profile of actions. For all $a\in A$ and all $x\in X,$
principal $j$'s (von Neumann Morgenstern expected) utility is $v_{j}(a,x)$
and agent $i$'s (von Neumann Morgenstern expected) utility is $u_{i}(a,x)$.

Each players's preference relation satisfies the von Neumann-Morgenstern
expected utility property. Unless otherwise specified, I assume that each
agent's utility function is \emph{additively separable} with respect to
principals' actions. That is, for all $a\in A$ and all $x\in X,$ agent $i$'s
utility takes the form of 
\begin{equation}
u_{i}(a,x)=\sum_{k=1}^{J}u_{ik}(a_{k},x).  \label{agent_payoff}
\end{equation}

Notably, (\ref{agent_payoff}) allows a quasilinear representation of an
agent's preference relation. For example, let $a_{j}=(y_{j},t_{j}),$ where $%
y_{j}$ is a non-monetary (multidimensional) component of principal $j$'s
action and $t_{j}=(t_{0j},t_{1j},\ldots ,t_{Ij})$ is a profile of monetary
transfers with $t_{0j}$ for principal $j$ himself and $t_{ij}$ for agent $i$
for all $i\in \mathcal{I}$. Then, (\ref{agent_payoff}) allows us to express
agent $i$'s utility as $u_{i}(a,x)=\sum_{k=1}^{J}h_{ik}(y_{k},x)+T_{i}$ with 
$T_{i}=\sum_{k=1}^{J}t_{ik}+m_{i}$ for agent $i$'s initial amount of money $%
m_{i}\in 
\mathbb{R}
_{+}$.

While (\ref{agent_payoff}) imposes the separability on the agent's utility
function, it is quite general in terms of the nature of actions, how
principals' actions affect a player's utility, and how a player's utility
depends on all agents' types. The equivalence between dominant-strategy
implementation and Bayesian implementation in Gershkov, et al. (2013) is
based on the restrictions that (i) agents' utility functions are linear,
(ii) each agent's type has only \emph{private value}, and (iii) a
principal's action is $a=(n,m)$, where an alternative $n$ belongs to a
discrete and finite set $\{1,\ldots ,N\}$ and $m$ is a vector of monetary
transfers, one for each agent. Given that the agent's utility environment in
my paper allows for \emph{interdependent values} and generality of a
principal's action, working with dominant strategy incentive compatible
direct mechanisms for the characterization of equilibrium allocations loses
its generality even with linear utility functions and monetary transfers.%
\footnote{%
One may consider ex-post incentive compatible direct mechanisms for the case
of interdependent types. However they are generally not equivalent to
Bayesian incentive compatible direct mechanisms as shown in Gershkov, et al
(2013).}

A major departure from the literature on competing mechanisms is to allow a
principal to send a message to himself given his mechanism. Specifically,
each principal $j$ can offer a \emph{general} \emph{mechanism} that
determines his action contingent on agents' messages and his own message.
Let $M_{ij}$ be the set of messages available for agent $i\in \mathcal{I}.$
Let $M_{0j}$ be the set of messages available for principal $j$ himself. Let 
$M_{j}=\times _{k=0}^{I}M_{kj}.$ A mechanism is then a measurable mapping 
\begin{equation*}
\gamma _{j}:M_{j}\rightarrow A_{j}.
\end{equation*}%
Let $\Gamma _{j}$ be the set of mechanisms available for principal $j$. Let $%
\Gamma =\Gamma _{1}\times \cdots \times \Gamma _{J}.$

The class of mechanisms proposed in this paper is more general than a
standard mechanism that specifies a principal's action contingent on agents'
messages only. In terms of our formulation, a standard mechanism is the same
as a mechanism $\gamma _{j}:M_{j}\rightarrow A_{j}$ such that for all $%
m_{0j},m_{0j}^{\prime }\in M_{0j}$ and all $(m_{1j},\ldots ,m_{Ij})\in
\times _{k\neq 0}M_{kj}$,%
\begin{equation}
\gamma _{j}(m_{0j},m_{1j},\ldots ,m_{Ij})=\gamma _{j}(m_{0j}^{\prime
},m_{1j},\ldots ,m_{Ij}).  \label{restriction1}
\end{equation}%
Let $\overline{\Gamma }_{j}$ be the set of mechanisms satisfying (\ref%
{restriction1}) for each $j\in \mathcal{J}$. Clearly $\overline{\Gamma }_{j}$
is a strict subset of $\Gamma _{j}$. Standard competing mechanism games
(e.g., Epstein and Peters (1999), Han (2006), Yamashita (2010) among many)
allow principal $j$ to offer only a mechanism in $\overline{\Gamma }_{j}$.
For notational simplicity, I drop the message set $M_{0j}$ for any mechanism
in $\overline{\Gamma }_{j}$ and denote a mechanism in $\overline{\Gamma }%
_{j} $ by a measurable mapping from $M_{1j}\times \cdots \times M_{Ij}$ into 
$A_{j}.$

Mechanisms in $\Gamma $ or $\overline{\Gamma }$ are quite general in terms
of the degree and nature of communication that they allow. For example, they
are sufficiently general that there exists an embedding $\eta _{j}$ from the
set of individually BIC direct mechanisms for principal $j$ defined in
Section 3 into $\overline{\Gamma }_{j}$. $\overline{\Gamma }$ is also large
enough to include mechanisms homeomorphic to deviator-reporting mechanisms
and deviator's DM (direct mechanism) reporting mechanisms defined in
Sections 4 and 5.

For a contract in Szentes (2010), there exists a strategically equivalent
mechanism in $\Gamma _{j}$. Principal $j$'s contract proposed by Szentes
(2010) is a measurable mapping $h_{j}:M_{1j}\times \cdots \times
M_{Ij}\rightarrow \mathcal{P}(A_{j})$ in our environment, where $\mathcal{P}%
(A_{j})$ is the set of all subsets of $A_{j}.$ In this contract, a profile
of agents' messages $(m_{1j},\ldots ,m_{Ij})$ first determines $%
h_{j}(m_{1j},\ldots ,m_{Ij}),$ a subset of $A_{j},$ from which principal $j$
subsequently chooses his action. For such a contract, I can construct a
strategically equivalent mechanism $\gamma _{j}:M_{j}\rightarrow A_{j}$. Let 
$m_{0j}$ be a measurable mapping from $h_{j}(M_{1j}\times \cdots \times
M_{Ij})$ (i.e., the image set of $h_{j}$) into $A_{j}$ with the property
that for all $h_{j}(m_{1j},\ldots ,m_{Ij})\in h_{j}(M_{1j}\times \cdots
\times M_{Ij})$, $m_{0j}(h_{j}(m_{1j},\ldots ,m_{Ij}))\in
h_{j}(m_{1j},\ldots ,m_{Ij}).$ Let $M_{0j}$ be the set of all such
measurable mappings. Let us define a mechanism $\gamma _{j}:M_{j}\rightarrow
A_{j}$ such that for a general contract $h_{j}$ and all $(m_{0j},m_{1j},%
\ldots ,m_{Ij})\in M_{j}$%
\begin{equation*}
\gamma _{j}(m_{0j},m_{1j},\ldots ,m_{Ij})=m_{0j}\left( h_{j}\left(
m_{1j},\ldots ,m_{Ij}\right) \right) .
\end{equation*}%
In a contract $h_{j},$ agents' messages determine a subset of actions and
then principal $j$ chooses an action from the subset. In the corresponding
mechanism $\gamma _{j}$, principal $j$ reports his contingency plan $m_{0j}$
to himself, at the same time agents send messages to him, as to what action
to choose as a function of the subset of actions determined by agents'
messages. Therefore, the set of contracts in Szentes (2010) can be nested
into $\Gamma $.

A general competing mechanism game with $\Gamma $ starts when principals
simultaneously offer their mechanisms from $\Gamma .$ Given a profile of
mechanisms $\gamma =(\gamma _{1},\ldots ,\gamma _{J})$, each agent sends a
message to every principal and every principal sends a message only to
himself at the same time. When principals and agents send messages, each
principal $j$'s mechanism is only observed by agents but not by the other
principals. According to each principal $j$'s mechanism $\gamma _{j}$, a
profile of messages $m_{j}=(m_{0j},m_{1j},\ldots ,m_{Ij})$ determines his
random action as $\gamma _{j}(m_{j})\in A_{j}$. Finally, utilities are
realized.

One may consider an alternative game where a principal inputs his message
after observing the other principals' mechanisms. In terms of characterizing
equilibrium allocations, it does not matter whether or not the other
principals' mechanisms are observable by a principal when he inputs his
message in his mechanism given three or more agents. The reason is that the
message spaces can be sufficiently general that (i) agents are able to
describe the other principals' mechanisms with their messages and (ii) a
principal can input a message in his mechanism that describes a contingency
plan with respect to all possible messages that agents could send, as
explained earlier.

\section{Notion of Equilibrium\label{sec_equilibrium_notion}}

I consider a pure-strategy equilibrium, as in Yamashita (2010), in the sense
that principals use pure strategies for their mechanism choices. In this
paper, mixed communication strategies are however allowed to such an extent
that the existence of an equilibrium can be addressed, whereas Yamashita
(2010) considers only pure communication strategies.

Let $c_{ij}:\Gamma \times X_{i}\rightarrow \Delta (M_{ij})$ denote agent $i$%
's mixed strategy of communicating with principal $j$. Let $c_{0j}:\Gamma
_{j}\rightarrow \Delta (M_{0j})$ be principal $j$'s mixed strategy of
communicating with himself. For all $\gamma \in \Gamma $ and all $x\in X$,
let 
\begin{multline*}
g_{j}(\gamma _{j},c_{0j}(\gamma _{j}),c_{1j}\left( \gamma ,x_{1}\right)
,\ldots ,c_{Ij}\left( \gamma ,x_{I}\right) ) \\
:=\int_{M_{0j}}\int_{M_{1j}}\cdots \int_{M_{Ij}}\gamma
_{j}(m_{0j},m_{1j},\ldots ,m_{Ij})dc_{0j}(\gamma _{j})dc_{1j}(\gamma
,x_{1})\cdots dc_{Ij}(\gamma ,x_{I}).
\end{multline*}%
Given $\gamma =(\gamma _{1},\ldots ,\gamma _{J})\in \Gamma ,$ $g_{j}(\gamma
_{j},c_{0j}(\gamma _{j}),c_{1j}\left( \gamma ,x_{1}\right) ,\ldots
,c_{Ij}\left( \gamma ,x_{I}\right) )\in A_{j}$ is the probability
distribution over $A_{j}$ conditional on $x=(x_{1},\ldots ,x_{I})$ that is
induced by (i) the strategies of communicating with principal $j$, $%
c_{j}(\gamma )=(c_{0j}(\gamma _{j}),c_{1j}\left( \gamma ,\cdot \right)
,\ldots ,c_{Ij}\left( \gamma ,\cdot \right) )$ and (ii) principal $j$'s
mechanism $\gamma _{j}$. Then, 
\begin{equation}
g_{j}(\gamma _{j},c\left( \gamma \right) )=g_{j}(\gamma _{j},c_{0j}(\gamma
_{j}),c_{1j}\left( \gamma ,\cdot \right) ,\ldots ,c_{Ij}\left( \gamma ,\cdot
\right) )  \label{direct_mech}
\end{equation}%
is a mapping from $X\rightarrow A_{j}$, which is a \emph{direct mechanism}
induced by the communication strategies $c_{j}(\gamma )$ and principal $j$'s
mechanism $\gamma _{j}$.\footnote{%
Suppose that there is only one principal in a model. Given a mechanism $%
\gamma _{1}$ offered by principal 1 (i.e., single principal), a profile of
continuation-equilibrium communication strategies $c_{1}\left( \gamma
_{1}\right) $ induces a Bayesian incentive compatible direct mechanism $%
g_{1} $ according to (\ref{direct_mech}) in which only agents send type
messages and the principal does not send a message to himself. Therefore,
the standard Revelation Principle holds. This implies that in a model of a
single principal, the class of mechanisms considered in Szentes (2010) and
my paper is no different from the class of standard mechanisms in terms of
the set of equilibrium allocations that it generates.\label{foot_RP}} Let us
denote principal $j$'s direct mechanism by $\pi _{j}:X\rightarrow A_{j}$.
Let $\Pi _{j}$ be the set of all possible (measurable) mappings $\pi
_{j}:X\rightarrow A_{j}.$

For $\gamma \in \Gamma ,$ agent $i$'s (ex-ante) expected utility is
expressed as%
\begin{equation*}
\mathbb{E}_{x}\left[ u_{i}(g_{1}(\gamma _{1},c_{1}(\gamma ,x)),\ldots
,g_{J}(\gamma _{J},c_{J}(\gamma ,x)),x)\right] =\sum_{k=1}^{J}\mathbb{E}_{x}%
\left[ u_{ik}(g_{k}(\gamma _{k},c_{k}\left( \gamma ,x\right) ),x)\right] ,
\end{equation*}%
where $\mathbb{E}_{x}\left[ \cdot \right] \mathbb{\ }$the expectation
operator based on the probability distribution $F$ over $X$ for all $k\in 
\mathcal{J}$. For $\gamma \in \Gamma $, principal $j$'s (ex-ante) expected
utility is denoted by 
\begin{equation*}
\mathbb{E}_{x}\left[ v_{j}(g_{1}(\gamma _{1},c_{1}(\gamma ,x)),\ldots
,g_{J}(\gamma _{J},c_{J}(\gamma ,x)),x)\right] .
\end{equation*}

\begin{definition}
\label{def_cont_eq}For all $\gamma \in \Gamma $, $c\left( \gamma \right)
=\left( c_{1}(\gamma ),\ldots ,c_{J}(\gamma )\right) $ is said to be a \emph{%
continuation equilibrium} given $\gamma $ if

\begin{enumerate}
\item[(i)] for all $i\in \mathcal{I}$ and all her communication strategies $%
(c_{i1}^{\prime }(\gamma ,\cdot ),\ldots ,c_{iJ}^{\prime }(\gamma ,\cdot ))$ 
\begin{multline}
\mathbb{E}_{x}\left[ u_{i}(\left[ g_{k}(\gamma _{k},c_{0k}(\gamma
_{k}),c_{1k}\left( \gamma ,x_{1}\right) ,\ldots ,c_{Ik}\left( \gamma
,x_{I}\right) )\right] _{k=1}^{J},x)\right] \geq \\
\mathbb{E}_{x}\left[ u_{i}(\left[ g_{k}(\gamma _{k},c_{0k}(\gamma
_{k}),c_{ik}^{\prime }\left( \gamma ,x_{i}\right) ,c_{-ik}\left( \gamma
,x_{-i}\right) )\right] _{k=1}^{J},x)\right] ,  \notag
\end{multline}%
where $c_{-ik}\left( \gamma ,x_{-i}\right) =\left( c_{\ell k}\left( \gamma
,x_{-\ell }\right) \right) _{\ell \neq i}$ and

\item[(ii)] for all $j\in \mathcal{J}$ and all his communication strategies $%
c_{0j}^{\prime }(\gamma )$,%
\begin{multline}
\mathbb{E}_{x}\left[ v_{j}(g_{j}(\gamma _{j},c_{j}(\gamma ,x)),g_{-j}(\gamma
_{-j},c_{-j}(\gamma ,x)),x)\right] \geq \\
\mathbb{E}_{x}\left[ v_{j}(g_{j}(\gamma _{j},c_{0j}^{\prime }(\gamma
),c_{-0,j}(\gamma ,x)),g_{-j}(\gamma _{-j},c_{-j}(\gamma ,x)),x)\right] ,
\label{principal_eq_payoff}
\end{multline}%
where $c_{-0,j}(\gamma ,x)=(c_{ij}(\gamma ,x_{i}))_{i\in \mathcal{I}}$ and $%
g_{-j}(\gamma _{-j},c_{-j}(\gamma ,x))=(g_{k}(\gamma _{k},c_{k}(\gamma
,x)))_{k\neq j}$.
\end{enumerate}
\end{definition}

Definition \ref{def_cont_eq} is based on ex-ante expected utilities. When
the agent's utility function is additively separable with respect to
principals' actions, condition (i) in Definition \ref{def_cont_eq} can be
equivalently formulated with agents' interim expected utilities as follows:
For all $i\in \mathcal{I}$, a.e. $x_{i}\in X_{i}$ and all $(\mu _{i1},\ldots
,\mu _{iJ})\in \Delta (M_{i1})\times \cdots \times \Delta (M_{iJ}),$%
\begin{multline}
\sum_{k=1}^{J}\mathbb{E}_{x_{-i}}\left[ u_{ik}(g_{k}(\gamma
_{k},c_{0k}(\gamma _{k}),c_{1k}\left( \gamma ,x_{1}\right) ,\ldots
,c_{Ik}\left( \gamma ,x_{I}\right) ),x_{-i},x_{i})|x_{i}\right] \geq \\
\sum_{k=1}^{J}\mathbb{E}_{x_{-i}}\left[ u_{ik}(g_{k}(\gamma
_{k},c_{0k}(\gamma _{k}),\mu _{ik},c_{-ik}\left( \gamma ,x_{-i}\right)
),x_{-i},x_{i})|x_{i}\right] ,  \label{agent_eq_payoff}
\end{multline}%
where $\mathbb{E}_{x_{-i}}\left[ \cdot |x_{i}\right] $ is the expectation
operator based on the probability distribution over $X_{-i}$ conditional on $%
x_{i}\in X_{i}$. Let $\mathcal{C}(\gamma )$ be the set of all continuation
equilibria given $\gamma .$ For all $\gamma \in \Gamma $, let $\mathcal{C}%
_{j}(\gamma )$ be the projection of $\mathcal{C}(\gamma )$ onto the space of
strategies of communicating with principal $j.$ For all $\gamma \in \Gamma $%
, let $\mathcal{C}_{-j}(\gamma )$ be the projection of $\mathcal{C}(\gamma )$
onto the space of strategies of communicating with principals except for $j.$

\begin{definition}
\label{def_eq}$(\gamma ^{\ast },c^{\ast })$ is a \emph{pure-strategy perfect
Bayesian equilibrium} (henceforth simply PBE) of a general competing
mechanism game with $\Gamma $ if (i) $c^{\ast }(\gamma )\in \mathcal{C}%
\left( \gamma \right) $ for all $\gamma \in \Gamma $ and (ii) for all $j\in 
\mathcal{J}$ and all $\gamma _{j}\in \Gamma _{j}$,%
\begin{multline}
\mathbb{E}_{x}\left[ v_{j}\left( g_{j}\left( \gamma _{j}^{\ast },c_{j}^{\ast
}(\gamma ^{\ast },x)\right) ,g_{-j}\left( \gamma _{-j}^{\ast },c_{-j}^{\ast
}(\gamma ^{\ast },x)\right) ,x\right) \right] \geq \\
\mathbb{E}_{x}\left[ v_{j}\left( g_{j}\left( \gamma _{j},c_{j}^{\ast
}(\gamma _{j},\gamma _{-j}^{\ast },x)\right) ,g_{-j}\left( \gamma
_{-j}^{\ast },c_{-j}^{\ast }(\gamma _{j},\gamma _{-j}^{\ast },x)\right)
,x\right) \right] .  \label{original_eq_cond}
\end{multline}
\end{definition}

Yamashita (2010) and Xiong (2013) employ the notions of sequential
equilibrium and subgame perfect Nash equilibrium respectively with finite
types, whereas Szentes (2015) employs the notion of perfect Bayesian
equilibrium in the model of contractible contracts. All these notions,
including one defined in Definition \ref{def_eq} above, are in fact
identical except for the extent to which mixed strategies are allowed. A
common requirement of these equilibrium notions is that principals have no
incentive to deviate as long as there exists a continuation equilibrium upon
deviation where the deviating principal cannot gain.

A strongly robust equilibrium (Epstein and Peters (1999) and Han (2007)) is
an alternative notion of equilibrium. In our general competing mechanism
game, strongly robust PBE can be defined as follows.

\begin{definition}
\label{def_strong_eq}$(\gamma ^{\ast },c^{\ast })$ is a \emph{strongly
robust PBE} of a general competing mechanism game with $\Gamma $ if (i) it
is a PBE and (ii) for all $j\in \mathcal{J}$, all $\gamma _{j}\in \Gamma
_{j} $, and all $[c_{j}(\gamma _{j},\gamma _{-j}^{\ast }),c_{-j}(\gamma
_{j},\gamma _{-j}^{\ast })]\in \mathcal{C}_{j}(\gamma _{j},\gamma
_{-j}^{\ast }),$%
\begin{multline}
\mathbb{E}_{x}\left[ v_{j}\left( g_{j}\left( \gamma _{j}^{\ast },c_{j}^{\ast
}(\gamma ^{\ast },x)\right) ,g_{-j}\left( \gamma _{-j}^{\ast },c_{-j}^{\ast
}(\gamma ^{\ast },x)\right) ,x\right) \right] \geq \\
\mathbb{E}_{x}\left[ v_{j}\left( g_{j}\left( \gamma _{j},c_{j}(\gamma
_{j},\gamma _{-j}^{\ast },x)\right) ,g_{-j}\left( \gamma _{-j}^{\ast
},c_{-j}(\gamma _{j},\gamma _{-j}^{\ast },x)\right) ,x\right) \right] .
\label{strong_eq_cond}
\end{multline}
\end{definition}

Condition (ii) in Definition \ref{def_strong_eq} is equivalent to the
following condition: For all $j$ and all $\gamma _{j}\in \Gamma _{j}$%
\begin{multline*}
\mathbb{E}_{x}\left[ v_{j}\left( g_{j}\left( \gamma _{j}^{\ast },c_{j}^{\ast
}(\gamma ^{\ast },x)\right) ,g_{-j}\left( \gamma _{-j}^{\ast },c_{-j}^{\ast
}(\gamma ^{\ast },x)\right) ,x\right) \right] \geq \\
\sup_{\lbrack c_{j}(\gamma _{j},\gamma _{-j}^{\ast }),c_{-j}(\gamma
_{j},\gamma _{-j}^{\ast })]\in \mathcal{C}(\gamma _{j},\gamma _{-j}^{\ast })}%
\mathbb{E}_{x}\left[ v_{j}\left( g_{j}\left( \gamma _{j},c_{j}(\gamma
_{j},\gamma _{-j}^{\ast },x)\right) ,g_{-j}\left( \gamma _{-j}^{\ast
},c_{-j}(\gamma _{j},\gamma _{-j}^{\ast },x)\right) ,x\right) \right]
\end{multline*}%
A principal's equilibrium utility in a strongly robust PBE is compared to
his utility based on continuation equilibrium strategies of communicating
with all principals that are \emph{best} for him upon his deviation to any
mechanism. The literature has not identified the set of strongly robust PBE
allocations nor the greatest lower bound of each principal's strongly robust
PBE utility even in terms of mechanisms allowed in the game.

I propose a new notion of equilibrium. For all $\gamma \in \Gamma $, all $%
c_{-j}\left( \gamma \right) \in \mathcal{C}_{-j}(\gamma )$, let us define
the set of strategies of communicating with principal $j$ as follows 
\begin{equation*}
\mathcal{C}_{j}(\gamma |c_{-j}\left( \gamma \right) ):=\left\{ c_{j}(\gamma
):\left( c_{j}\left( \gamma \right) ,c_{-j}\left( \gamma \right) \right) \in 
\mathcal{C}(\gamma )\right\} .
\end{equation*}

\begin{definition}
\label{Def_robust_eq}$(\gamma ^{\ast },c^{\ast })$ is a \emph{robust PBE} of
a general competing mechanism game with $\Gamma $ if

\begin{enumerate}
\item[(i)] $c^{\ast }(\gamma )\in \mathcal{C}(\gamma )$ for all $\gamma \in
\Gamma $ and

\item[(ii)] for all $j\in \mathcal{J}$, all $\gamma _{j}\in \Gamma _{j}$,
and all $c_{j}(\gamma _{j},\gamma _{-j}^{\ast })\in \mathcal{C}_{j}(\gamma
_{j},\gamma _{-j}^{\ast }|c_{-j}^{\ast }(\gamma _{j},\gamma _{-j}^{\ast }))$%
, 
\begin{multline}
\mathbb{E}_{x}\left[ v_{j}\left( g_{j}\left( \gamma _{j}^{\ast },c_{j}^{\ast
}(\gamma ^{\ast },x)\right) ,g_{-j}\left( \gamma _{-j}^{\ast },c_{-j}^{\ast
}(\gamma ^{\ast },x)\right) ,x\right) \right] \geq \\
\mathbb{E}_{x}\left[ v_{j}\left( g_{j}\left( \gamma _{j},c_{j}(\gamma
_{j},\gamma _{-j}^{\ast },x)\right) ,g_{-j}\left( \gamma _{-j}^{\ast
},c_{-j}^{\ast }(\gamma _{j},\gamma _{-j}^{\ast },x)\right) ,x\right) \right]
.  \label{robust_eq_cond}
\end{multline}
\end{enumerate}
\end{definition}

A PBE $(\gamma ^{\ast },c^{\ast })$ is robust if any principal $j$ cannot
gain regardless of a profile of agents' continuation equilibrium strategies $%
c_{j}(\gamma _{j},\gamma _{-j}^{\ast })$ of communicating with $j$ upon his
deviation to any general mechanism $\gamma _{j}$, \emph{taking as given} the
profile of their continuation equilibrium strategies $c_{-j}^{\ast }(\gamma
_{j},\gamma _{-j}^{\ast })$ of communicating with non-deviating principals.

The notion of robust PBE is weaker than the notion of strongly robust PBE
(Epstein and Peters (1999), Han (2007)) in which a principal cannot gain in
all continuation equilibria agents may play upon his deviation. It is first
interesting to identify the utility environments where the notion of robust
PBE coincides with that of strongly robust PBE.

\begin{definition}
We say that principal $j$'s utility function satisfies the \textquotedblleft 
\emph{no action externalities}\textquotedblright\ condition if 
\begin{equation}
v_{j}(a_{j},a_{-j},x)=v_{j}(a_{j},a_{-j}^{\prime },x),\text{ }\forall
a_{j}\in A_{j},\forall (a_{-j},a_{-j}^{\prime })\in \left( A_{-j}\right)
^{2},\forall x\in X.  \label{no_ext}
\end{equation}
\end{definition}

(\ref{no_ext}) implies that principal $j$ is indifferent between any action
profiles that the other principals may choose, regardless of his own action
and the agents' types.

\begin{proposition}
Suppose that each agent's utility function is additively separable with
respect to principals' actions and that each principal's utility function
satisfies the \textquotedblleft no action externalities\textquotedblright\
condition. Then, the notion of robust PBE coincides with that of strongly
robust PBE.
\end{proposition}

\begin{proof}
Consider a robust PBE $(\gamma ^{\ast },c^{\ast })$. Because it is robust,
it satisfies Definition \ref{Def_robust_eq}.(ii). Because each agent's
utility function is additively separable with respect to principals'
actions, we have that 
\begin{eqnarray}
\mathcal{C}_{j}(\gamma _{j},\gamma _{-j}|c_{-j}(\gamma _{j},\gamma _{-j}))
&=&\mathcal{C}_{j}(\gamma _{j},\gamma _{-j}|c_{-j}^{\prime }(\gamma
_{j},\gamma _{-j}))\text{, }  \label{cont_eq_other_prin} \\
\forall j &\in &\mathcal{J}\text{, }\forall (\gamma _{j},\gamma _{-j})\in
\Gamma ,\forall \left( c_{-j}(\gamma _{j},\gamma _{-j}),c_{-j}^{\prime
}(\gamma _{j},\gamma _{-j})\right) \in \mathcal{C}_{-j}(\gamma _{j},\gamma
_{-j})^{2}  \notag
\end{eqnarray}%
Furthermore, because each principal's utility function satisfies the
\textquotedblleft no action externalities\textquotedblright\ condition, we
have that 
\begin{multline}
\mathbb{E}_{x}\left[ v_{j}\left( g_{j}\left( \gamma _{j},c_{j}(\gamma
,x)\right) ,g_{-j}\left( \gamma _{-j},c_{-j}(\gamma ,x)\right) ,x\right) %
\right] =  \label{prin_payoff} \\
\mathbb{E}_{x}\left[ v_{j}\left( g_{j}\left( \gamma _{j},c_{j}(\gamma
,x)\right) ,g_{-j}\left( \gamma _{-j},c_{-j}^{\prime }(\gamma ,x)\right)
,x\right) \right] , \\
\forall j\in \mathcal{J}\text{, }\forall \gamma =(\gamma _{j},\gamma
_{-j})\in \Gamma ,\forall \left( c_{-j}(\gamma ),c_{-j}^{\prime }(\gamma
)\right) \in \mathcal{C}_{-j}(\gamma )^{2}, \\
\forall c_{j}(\gamma ,x)\in \mathcal{C}_{j}(\gamma |c_{-j}(\gamma ))=%
\mathcal{C}_{j}(\gamma |c_{-j}^{\prime }(\gamma )),
\end{multline}%
where $\mathcal{C}_{j}(\gamma |c_{-j}(\gamma ))=\mathcal{C}_{j}(\gamma
|c_{-j}^{\prime }(\gamma ))$ is due to (\ref{cont_eq_other_prin}).

Then, Definition \ref{Def_robust_eq}.(ii) can be rewritten as follows: For
all $j\in \mathcal{J}$, all $\gamma _{j}\in \Gamma _{j}$, all $c_{-j}(\gamma
_{j},\gamma _{-j}^{\ast })\in \mathcal{C}_{-j}(\gamma _{j},\gamma
_{-j}^{\ast }),$ and all $c_{j}(\gamma _{j},\gamma _{-j}^{\ast })\in 
\mathcal{C}_{j}(\gamma _{j},\gamma _{-j}^{\ast }|c_{-j}^{\ast }(\gamma
_{j},\gamma _{-j}^{\ast }))=\mathcal{C}_{j}(\gamma _{j},\gamma _{-j}^{\ast
}|c_{-j}(\gamma _{j},\gamma _{-j}^{\ast }))$, 
\begin{multline*}
\mathbb{E}_{x}\left[ v_{j}\left( g_{j}\left( \gamma _{j}^{\ast },c_{j}^{\ast
}(\gamma ^{\ast },x)\right) ,g_{-j}\left( \gamma _{-j}^{\ast },c_{-j}^{\ast
}(\gamma ^{\ast },x)\right) ,x\right) \right] \geq \\
\mathbb{E}_{x}\left[ v_{j}\left( g_{j}\left( \gamma _{j},c_{j}(\gamma
_{j},\gamma _{-j}^{\ast },x)\right) ,g_{-j}\left( \gamma _{-j}^{\ast
},c_{-j}^{\ast }(\gamma _{j},\gamma _{-j}^{\ast },x)\right) ,x\right) \right]
= \\
\mathbb{E}_{x}\left[ v_{j}\left( g_{j}\left( \gamma _{j},c_{j}(\gamma
_{j},\gamma _{-j}^{\ast },x)\right) ,g_{-j}\left( \gamma _{-j}^{\ast
},c_{-j}(\gamma _{j},\gamma _{-j}^{\ast },x)\right) ,x\right) \right] ,
\end{multline*}%
where the inequality is simply Definition \ref{Def_robust_eq}.(ii) and the
equality holds because of (\ref{prin_payoff}). Therefore, a robust PBE is
strongly robust. On the other hand, any strongly robust is clearly robust.
\end{proof}

\bigskip

Note that the no action externalities condition for each principal's utility
function is not sufficient for the equivalence between robust PBE and
strongly robust PBE. The reason is that the set of continuation equilibrium
strategies of agents' communication with principal $j$ (i.e., $\mathcal{C}%
_{j}(\gamma _{j},\gamma _{-j}|c_{-j}(\gamma _{j},\gamma _{-j}))$) may depend
on the continuation strategies of their communication with the other
principals, $c_{-j}(\gamma _{j},\gamma _{-j})$. The separability condition
for each agent's utility function breaks down such dependence and
establishes (\ref{cont_eq_other_prin}). This, together with the no action
externalities condition for each principal's utility function, leads to the
equivalence between robust PBE and strongly robust PBE.

The notion of robust PBE is stronger than that of PBE because an allocation
can be supported in a PBE as long as there exists a continuation equilibrium
in which a principal cannot gain upon deviation. Consider the following
condition for each principal's payoff function.

\begin{definition}
We say that principal $j$'s payoff function satisfies the \textquotedblleft 
\emph{pure externalities}\textquotedblright\ condition if 
\begin{equation}
v_{j}(a_{j},a_{-j},x)=v_{j}(a_{j}^{\prime },a_{-j},x),\text{ }\forall \left(
a_{j},a_{j}^{\prime }\right) \in \left( A_{j}\right) ^{2},\forall a_{-j}\in
A_{-j},\forall x\in X.  \label{pure_ext}
\end{equation}
\end{definition}

That is, principal $j$ is indifferent between any actions he may choose,
regardless of the action profile that other principals choose and the
agents' types.

\begin{proposition}
Suppose that each principal's utility function satisfies the
\textquotedblleft pure externalities\textquotedblright\ condition. Then, the
notion of robust PBE coincides with that of PBE.
\end{proposition}

\begin{proof}
Consider a PBE $(\gamma ^{\ast },c^{\ast })$. If each principal's utility
function satisfies the \textquotedblleft pure
externalities\textquotedblright\ condition, the following condition is
satisfied: 
\begin{multline}
\mathbb{E}_{x}\left[ v_{j}\left( g_{j}\left( \gamma _{j}^{\ast },c_{j}^{\ast
}(\gamma ^{\ast },x)\right) ,g_{-j}\left( \gamma _{-j}^{\ast },c_{-j}^{\ast
}(\gamma ^{\ast },x)\right) ,x\right) \right] =  \label{indifference} \\
\mathbb{E}_{x}\left[ v_{j}\left( g_{j}\left( \gamma _{j},c_{j}(\gamma
_{j},\gamma _{-j}^{\ast },x)\right) ,g_{-j}\left( \gamma _{-j}^{\ast
},c_{-j}^{\ast }(\gamma _{j},\gamma _{-j}^{\ast },x)\right) ,x\right) \right]
, \\
\forall j\in \mathcal{J}\text{, }\forall \gamma _{j}\in \Gamma _{j},\forall
c_{j}(\gamma _{j},\gamma _{-j}^{\ast })\in \mathcal{C}_{j}(\gamma
_{j},\gamma _{-j}^{\ast }|c_{-j}^{\ast }(\gamma _{j},\gamma _{-j}^{\ast }))
\end{multline}%
(\ref{indifference}) implies that Definition \ref{Def_robust_eq}.(ii) is
satisfied. Therefore, $(\gamma ^{\ast },c^{\ast })$ is robust. On the other
hand, any robust PBE is clearly a PBE.Because it is robust, it satisfies
Definition \ref{def_eq}.(ii).
\end{proof}

\bigskip

For the equivalence between the notion of robust PBE and that of PBE, we
only need the \textquotedblleft pure externalities\textquotedblright\
condition for each principal's utility function, but no restrictions on the
agent's utility function. Even though each principal $j$ is indifferent
between any actions he may choose, he still faces a strategic choice of his
mechanism. The reason is that principal $j$'s mechanism $\gamma _{j}$
affects agents' continuation equilibrium strategies of communicating with
the other principals, $c_{-j}^{\ast }(\gamma _{j},\gamma _{-j}^{\ast },x)$,
which in turn affects principal $j$'s utility through $g_{-j}\left( \gamma
_{-j}^{\ast },c_{-j}^{\ast }(\gamma _{j},\gamma _{-j}^{\ast },x)\right) $.

\subsection{Variations of robust PBE}

One interesting variation of robust PBE is one where an agent's
communication with a non-deviating principal does not depend on any other
principal's mechanisms. That is, 
\begin{equation}
c_{j}^{\ast }(\gamma _{j}^{\ast },\gamma _{-j})=c_{j}^{\ast }(\gamma
_{j}^{\ast },\gamma _{-j}^{\prime })\text{ }\forall j\in \mathcal{J}\text{, }%
\forall \gamma _{-j},\gamma _{-j}^{\prime }\in \Gamma _{-j}
\label{invariant}
\end{equation}%
given a robust PBE $(\gamma ^{\ast },c^{\ast })$.

\begin{definition}
A robust PBE $(\gamma ^{\ast },c^{\ast })$ is a \emph{market-information
invariant }if it satisfies (\ref{invariant}).
\end{definition}

In a market-information invariant robust PBE, each agent's communication
with a non-deviating principal only depends on her type and the principal's
equilibrium mechanism.

I also formulate a uniquely implementable robust PBE as follows.

\begin{definition}
A robust PBE $(\gamma ^{\ast },c^{\ast })$ is uniquely implementable if, for
all $j\in \mathcal{J}$, all $k\neq j$ and all $\gamma _{j}\in \Gamma _{j}$ 
\begin{equation}
g_{k}\left( \gamma _{k}^{\ast },c_{k}(\gamma _{j},\gamma _{-j}^{\ast
})\right) =g_{k}\left( \gamma _{k}^{\ast },c_{k}^{\prime }(\gamma
_{j},\gamma _{-j}^{\ast })\right) \in \Pi _{j}^{B}\text{, }\forall
c_{k}(\gamma _{j},\gamma _{-j}^{\ast }),c_{k}^{\prime }(\gamma _{j},\gamma
_{-j}^{\ast })\in \mathcal{C}_{k}\left( \gamma _{j},\gamma _{-j}^{\ast
}\right)  \label{unique_eq}
\end{equation}
\end{definition}

Fix a path in the game after principals' mechanism announcement: either
equilibrium path $\gamma _{j}=\gamma _{j}^{\ast }$ or off the equilirium
path following a principal $j$'s unilateral deviation to an arbitrary
general mechanism $\gamma _{j}\neq \gamma _{j}^{\ast }$. Given such a path,
any profile of continuation-equilibrium strategies of communicating with each%
$\ $non-deviating principal $k$ always induces the same BIC direct mechanism
for $k$. This implies that if a robust PBE is uniquely implementable, it is
also strongly robust. Furthermore, not only is there no continuation
equilibrium where a principal can gain off the equilibrium path following
his deviation, but there is a unique equilibrium allocation regardless of a
continuation-equilibrium agents play on the equilibrium path.

\section{Competing BIC direct mechanisms\label{sec:competing_BIC}}

In this section, I provide the characterizations of market-information
invariant robust PBE allocations and uniquely implementable robust PBE
allocations respectively. Let us first define incentive compatible direct
mechanisms. A profile of direct mechanisms $\pi =(\pi _{1},\ldots ,\pi _{J})$
is \emph{Bayesian incentive compatible} (BIC) if for all $i\in \mathcal{I}$,
a.e. $x_{i}\in X_{i}$, and all $x_{i1},\ldots ,x_{iJ}\in X_{i}$%
\begin{multline}
\mathbb{E}_{x_{-i}}\left[ u_{i}(\pi _{1}(x),\ldots ,\pi
_{J}(x),x_{-i},x_{i})|x_{i}\right] \geq \\
\mathbb{E}_{x_{-i}}\left[ u_{i}(\pi _{1}(x_{i1},x_{-i}),\ldots ,\pi
_{J}(x_{iJ},x_{-i}),x_{-i},x_{i})|x_{i}\right] .  \label{BIC_whole}
\end{multline}%
Let $\Pi ^{B}$ be the set of all BIC profiles of direct mechanisms.

Principal $j$'s direct mechanism is \emph{individually BIC} if for all $i\in 
\mathcal{I}$, a.e. $x_{i}\in X_{i}$, and $x_{ij}\in X_{i},$%
\begin{equation}
\mathbb{E}_{x_{-i}}\left[ u_{ij}(\pi _{j}\left( x_{i},x_{-i}\right)
,x_{-i},x_{i})|x_{i}\right] \geq \mathbb{E}_{x_{-i}}\left[ u_{ij}(\pi
_{j}\left( x_{ij},x_{-i}\right) ,x_{-i},x_{i})|x_{i}\right] .
\label{BIC_individual}
\end{equation}%
Let $\Pi _{j}^{B}$ be the set of all individually BIC direct mechanisms for
principal $j.$ Lemma \ref{lemma_BIC} below shows that the set of BIC
profiles of direct mechanisms is the Cartesian product of individually BIC
direct mechanisms. This is due to the additive separability of the agent's
utility function with respect to principals' actions.

\begin{lemma}
\label{lemma_BIC}$\Pi ^{B}=\Pi _{1}^{B}\times \cdots \times \Pi _{J}^{B}.$
\end{lemma}

\begin{proof}
Consider any BIC profile of direct mechanisms $\pi =(\pi _{1},\ldots ,\pi
_{J})\in \Pi ^{B}.$ Because $\pi $ is BIC as a whole, I have that for all $%
j\in \mathcal{J}$, all $i\in \mathcal{I}$, a.e. $x_{i}\in X_{i}$, and all $%
x_{ij}\in X_{i}$, 
\begin{multline*}
\sum_{k=1}^{J}\mathbb{E}_{x_{-i}}\left[ u_{ik}(\pi
_{k}(x),x_{-i},x_{i})|x_{i}\right] \geq \\
\mathbb{E}_{x_{-i}}\left[ u_{ij}(\pi _{j}(x_{ij},x_{i}),x_{-i},x_{i})|x_{i}%
\right] +\sum_{k\neq j}\mathbb{E}_{x_{-i}}\left[ u_{ik}(\pi
_{k}(x),x_{-i},x_{i})|x_{i}\right] ,
\end{multline*}%
The inequality above implies that for all $i\in \mathcal{I}$, a.e. $x_{i}\in
X_{i}$, and all $x_{ij}\in X_{i}$, 
\begin{equation*}
\mathbb{E}_{x_{-i}}\left[ u_{ij}(\pi _{j}\left( x_{i},x_{-i}\right)
,x_{-i},x_{i})|x_{i}\right] \geq \mathbb{E}_{x_{-i}}\left[ u_{ij}(\pi
_{j}\left( x_{ij},x_{-i}\right) ,x_{-i},x_{i})|x_{i}\right] .
\end{equation*}%
This means that $\pi _{j}$ is individually BIC and hence $\pi _{j}\in \Pi
_{j}^{B}$ for all $j\in \mathcal{J}$. Therefore, $\pi =(\pi _{1},\ldots ,\pi
_{J})\in \Pi _{1}^{B}\times \cdots \times \Pi _{J}^{B}.$

It is straightforward to show that any profile of individually BIC direct
mechanisms $\pi =(\pi _{1},\ldots ,\pi _{J})\in \Pi _{1}^{B}\times \cdots
\times \Pi _{J}^{B}$ is also a BIC profile of direct mechanisms, i.e., $\pi
=(\pi _{1},\ldots ,\pi _{J})\in \Pi ^{B}$.
\end{proof}

\bigskip

I consider the competing BIC direct mechanism game where each principal is
restricted to offer only a BIC direct mechanism in $\Pi _{j}^{B}$.

\begin{definition}
$\pi =\left( \pi _{1},\ldots \pi _{J}\right) \in \Pi ^{B}$ is a profile of
equilibrium BIC direct mechanisms in a (truth-telling) PBE of the competing
mechanism game with $\Pi ^{B}$if for all $j\in \mathcal{J}$, 
\begin{equation}
\mathbb{E}_{x}\left[ v_{j}(\pi _{j}(x),\pi _{-j}(x),x)\right] \geq \mathbb{E}%
_{x}\left[ v_{j}(\pi _{j}^{\prime }(x),\pi _{-j}(x),x)\right] \text{ }%
\forall \pi _{j}^{\prime }\in \Pi _{j}^{B}  \label{BIC_eq}
\end{equation}
\end{definition}

Now I show that the strategy space for each principal $j$ in the the
competing mechanism game with $\Pi ^{B}$ is non-empty, convex and compact.

\begin{proposition}
\label{Prop_compact_BIC}$\Pi _{j}^{B}$ is non-empty, convex, and compact for
all $j\in \mathcal{J}$.
\end{proposition}

\begin{proof}
$\Pi _{j}^{B}$ is non-empty because a constant direct mechanism is BIC. It
is also a compact convex set. It is straightforward to show that $\Pi
_{j}^{B}$ is a convex set because $\lambda \pi _{j}+(1-\lambda )\pi
_{j}^{\prime }$ is also BIC for all $\pi _{j},\pi _{j}^{\prime }\in \Pi
_{j}^{B}$ and all $\lambda \in \lbrack 0,1].$

$\Pi _{j}^{B}$ is compact if and only if it is totally bounded and complete
(See Theorem 3.11.2.2 in Kolmogorov and Fomin (1975)). To show that $\Pi
_{j}^{B}$ is totally bounded, it is enough to prove that $\Pi _{j}^{B}$ is
countably compact (See Theorem 3.11.2.1 in Kolmogorov and Fomin (1975)),
i.e., every sequence in $\Pi _{j}^{B}$ has a limit point in $\Pi _{j}^{B}$.
Suppose that $\Pi _{j}^{B}$ is not countably compact so that there is some
sequence $\left\{ \pi _{j}^{n}\right\} \in \Pi _{j}^{B}$ that has all its
limit point(s) in the complement of $\Pi _{j}^{B}$. The latter statement is
equivalent to the statement that there exists some sequence $\left\{ \pi
_{j}^{n}\right\} \in \Pi _{j}^{B}$ such that if $\left\{ \pi
_{j}^{n_{k}}\right\} \subseteq \left\{ \pi _{j}^{n}\right\} $ is any
arbitrary convergent subsequence with limit $\pi _{j}^{\ast }\notin \Pi
_{j}^{B}$. First, notice that $\left\{ \pi _{j}^{n_{k}}\right\} \in \Pi
_{j}^{B}$ since $\left\{ \pi _{j}^{n_{k}}\right\} \subseteq \left\{ \pi
_{j}^{n}\right\} \in \Pi _{j}^{B}$. Now, because $\pi _{j}^{\ast }\notin \Pi
_{j}^{B}$, there is some $i\in \mathcal{I}$ and some $D_{i}\subset X_{i}$
with positive measure such that $\pi _{j}^{\ast }\notin \Pi _{j}^{B}$ does
not satisfy (\ref{BIC_individual}) for all $x_{i}\in D_{i}$. Because $%
\left\{ \pi _{j}^{n_{k}}\right\} $ converges to $\pi _{j}^{\ast }$, there
exists some $N_{\circ }\in \mathbb{N}$ such that $\left\{ \pi
_{j}^{n_{k}}\right\} $ also violates (\ref{BIC_individual}) for all $%
x_{i}\in D_{i}$ and all $n_{k}\geq N_{\circ }$. This contradicts the fact
that $\left\{ \pi _{j}^{n_{k}}\right\} \in \Pi _{j}^{B}$. Therefore, $\Pi
_{j}^{B}$ must be countably compact, which implies that $\Pi _{j}^{B}$ is
totally bounded.

Finally, I show that $\Pi _{j}^{B}$ is complete. Recall that $\Pi _{j}$ is
the set of all possible (measurable) mappings $\pi _{j}:X\rightarrow A_{j}$.
Therefore, clearly $\Pi _{j}^{B}\subset \Pi _{j}\subseteq L^{p}\left(
X\right) $ with $1<p<\infty $. Because $L^{p}\left( X\right) $ is complete,
it suffices to show that $\Pi _{j}^{B}$ is closed. $\Pi _{j}^{B}$ is closed
if and only if every convergent sequence $\left\{ \pi _{j}^{k}\right\} $ in $%
\Pi _{j}^{B}$ has its limit $\pi _{j}^{\prime \ast }$in $\Pi _{j}^{B}.$ On
the contrary, suppose that there exists a convergent sequence $\left\{ \pi
_{j}^{k}\right\} $ in $\Pi _{j}^{B}$ with $\pi _{j}^{\prime \ast }\notin \Pi
_{j}^{B},$that is, for some $i\in \mathcal{I}$ and some $D_{i}^{\prime }\in
X_{i}$ with positive measure such that $\left\{ \pi _{j}^{k}\right\} $
violates (\ref{BIC_individual}) for all $x_{i}\in D_{i}^{\prime }$. Because $%
\left\{ \pi _{j}^{k}\right\} $ converges to $\pi _{j}^{\prime \ast }$, there
exists some $K\in \mathbb{N}$ such that $\left\{ \pi _{j}^{k}\right\} $ also
violates (\ref{BIC_individual}) for all $x_{i}\in D_{i}^{\prime }$ and all $%
k\geq K$. This contradicts the fact that $\left\{ \pi _{j}^{k}\right\} \in
\Pi _{j}^{B}.$ Thus, $\Pi _{j}^{B}$ is closed.
\end{proof}

\bigskip

Given Proposition \ref{Prop_compact_BIC}, the existence of a PBE BIC direct
mechanisms is straightforward.

\begin{corollary}
\label{corollary_existence}A (truth-telling) PBE exists in the competing
mechanism game with $\Pi ^{B}$
\end{corollary}

\begin{proof}
Because each $\Pi _{j}^{B}$ is non-empty compact convex (Proposition \ref%
{Prop_compact_BIC}) and $\mathbb{E}_{x}\left[ v_{j}(\pi (x),x)\right] $ is
linear in each $\Pi _{j}^{B}$, there exists a PBE in the competing mechanism
game with $\Pi ^{B}$ (Glicksberg 1952, Fudenberg and Tirole 1991).
\end{proof}

\bigskip

Furthermore, any (truth-telling) PBE in competing BIC direct mechanisms is
robust.

\begin{proposition}
A (truth-telling) PBE in the competing mechanism game with $\Pi ^{B}$ is $%
\emph{robust}$.
\end{proposition}

\begin{proof}
Suppose that principal $j$ deviates to $\pi _{j}^{\prime }$. Agents choose
any arbitrary profile $\left[ c_{1j}\left( \pi _{j}^{\prime },\pi _{j}\cdot
\right) ,\ldots ,c_{Ij}\left( \pi _{j}^{\prime },\pi _{j},\cdot \right) %
\right] $ of continuation equilibrium strategies of communicating wih $j$,
given their truthful type reporting to non-deviating principals. Then, $%
\left[ c_{1j}\left( \pi _{j}^{\prime },\pi _{j}\cdot \right) ,\ldots
,c_{Ij}\left( \pi _{j}^{\prime },\pi _{j},\cdot \right) \right] $ induces a
direct mechanism $\pi _{j}^{\prime \prime }$ such that for all $%
(x_{1},\ldots ,x_{I})\in X$ 
\begin{equation*}
\pi _{j}^{\prime \prime }(x_{1},\ldots ,x_{I})=\int_{X_{1}}\cdots
\int_{X_{I}}\pi _{j}^{\prime }(x_{1j},\ldots ,x_{Ij})dc_{1j}(\pi
_{j}^{\prime },\pi _{j},x_{1})\cdots dc_{Ij}(\pi _{j}^{\prime },\pi
_{j},x_{I})
\end{equation*}%
and $\pi _{j}^{\prime \prime }\in \Pi _{j}^{B}$. Because of (\ref{BIC_eq}),
I have that 
\begin{equation*}
\mathbb{E}_{x}\left[ v_{j}(\pi _{j}(x),\pi _{-j}(x),x)\right] \geq \mathbb{E}%
_{x}\left[ v_{j}(\pi _{j}^{\prime \prime }(x),\pi _{-j}(x),x)\right] .
\end{equation*}
\end{proof}

\subsection{Market-information invariant robust PBE allocations}

Now I show that the set of market-information invariant robust PBE
allocations in a competing mechanism game with $\Gamma $ is the set of
(truthful) PBE allocations in the competing direct mechanism game with $\Pi
^{B}$. To establish the equivalence, first recall that given $\gamma \in
\Gamma $, a profile of communication strategies $c(\gamma )$ induces a
direct mechanism $g_{j}(\gamma _{j},c_{0j}(\gamma _{k}),c_{1j}\left( \gamma
,\cdot \right) ,\ldots ,c_{Ij}\left( \gamma ,\cdot \right) )$ for each
principal $j\in \mathcal{J}$.

\begin{lemma}
\label{lemma_eq_BIC}Given any $\gamma \in \Gamma $ and any $c_{j}(\gamma
)=\left( c_{0j}(\gamma _{j}),c_{1j}\left( \gamma ,\cdot \right) ,\ldots
,c_{Ij}\left( \gamma ,\cdot \right) \right) \in \mathcal{C}_{j}(\gamma ),$ 
\begin{equation*}
g_{j}(\gamma _{j},c_{0j}(\gamma _{j}),c_{1j}\left( \gamma ,\cdot \right)
,\ldots ,c_{Ij}\left( \gamma ,\cdot \right) )\in \Pi _{j}^{B}\text{ for all }%
j\in \mathcal{J}
\end{equation*}
\end{lemma}

\begin{proof}
Fix a profile of mechanisms $\gamma \in \Gamma $ and a continuation
equilibrium $c_{j}(\gamma )\in \mathcal{C}_{j}(\gamma ).$ Suppose that $%
g_{j}(\gamma _{j},c_{0j}(\gamma _{j}),c_{1j}\left( \gamma ,\cdot \right)
,\ldots ,c_{Ij}\left( \gamma ,\cdot \right) )$ is not individually BIC. That
means that for some $j$, some $i$, some $B_{i}\subset X_{i}$ with $%
F(B_{i},X_{-i})>0$, and all $x_{i}\in B_{i},$ there exists $x_{i}^{\prime
}\neq x_{i}$ such that 
\begin{multline*}
\mathbb{E}_{x_{-i}}\left[ u_{ij}(g_{j}(\gamma _{j},c_{0j}(\gamma
_{j}),c_{-ij}\left( \gamma ,x_{-i}\right) ,c_{ij}\left( \gamma
,x_{i}^{\prime }\right) ),x_{-i},x_{i})|x_{i}\right] > \\
\mathbb{E}_{x_{-i}}\left[ u_{ij}(g_{j}(\gamma _{j},c_{0j}(\gamma
_{j}),c_{-ij}\left( \gamma ,x_{-i}\right) ,c_{ij}\left( \gamma ,x_{i}\right)
),x_{-i},x_{i})|x_{i}\right] .
\end{multline*}%
This contradicts that $c_{ij}\left( \gamma ,\cdot \right) $ is agent $i$'s
continuation-equilibrium strategy of communicating with principal $j$ given $%
\gamma $. Therefore, $g_{j}(\gamma _{j},c_{0j}(\gamma _{j}),c_{1j}\left(
\gamma ,\cdot \right) ,\ldots ,c_{Ij}\left( \gamma ,\cdot \right) )$ must be
individually BIC for all $j\in \mathcal{J}$.
\end{proof}

\bigskip

Let $\Pi _{B}^{\ast }$ be the set of all equilibrium allocations in
(truth-telling) PBEs of the competing direct mechanism with $\Pi ^{B}$: 
\begin{equation*}
\Pi _{B}^{\ast }:=\left\{ \pi \in \Pi ^{B}:\mathbb{E}_{x}\left[ v_{j}(\pi
_{j}(x),\pi _{-j}(x),x)\right] \geq \mathbb{E}_{x}\left[ v_{j}\left( \pi
_{j}^{\prime }(x),\pi _{-j}(x),x\right) \right] \text{ }\forall \pi
_{j}^{\prime }\in \Pi _{j}^{B}\right\}
\end{equation*}%
Because of Corollary \ref{corollary_existence}, $\Pi _{B}^{\ast }$ is
non-empty. Let $\Pi _{MR}^{\ast }$ be the set of all market-information
invariant robust PBE allocations in a general competing mechanism game with $%
\Gamma $: 
\begin{equation*}
\Pi _{MR}^{\ast }:=\left\{ 
\begin{array}{c}
\left( g_{k}\left( \gamma _{k}^{\ast },c_{01}^{\ast }(\gamma _{k}^{\ast
}),c_{1k}^{\ast }\left( \gamma ^{\ast },\cdot \right) ,\ldots ,c_{Ik}^{\ast
}\left( \gamma ^{\ast },\cdot \right) \right) \right) _{k\in \mathcal{J}}\in
\Pi ^{B}: \\ 
(\gamma ^{\ast },c^{\ast })\text{ is a market-information invariant robust
PBE.}%
\end{array}%
\right\} \text{.}
\end{equation*}

\begin{theorem}
\label{theorem2}$\Pi _{MR}^{\ast }=\Pi _{B}^{\ast }$
\end{theorem}

\begin{proof}
First of all, I prove that $\Pi _{MR}^{\ast }\supseteq \Pi _{B}^{\ast }.$
Fix $\pi ^{\ast }=\left( \pi _{1}^{\ast },\ldots \pi _{J}^{\ast }\right) \in
\Pi _{B}^{\ast }$. For each $j$ and $i$, let $M_{ij}$ be homeomorphic to $%
X_{i},$ that is, there is a homeomorphism $\phi _{ij}:X_{i}\rightarrow
M_{ij} $. In a competing general competing mechanism game with $\Gamma $,
each principal $j$ offers $\gamma _{j}^{\ast }$ such that $\gamma _{j}^{\ast
}\left( m_{0j},\phi _{1j}\left( x_{1}\right) ,\ldots ,\phi _{Ij}\left(
x_{I}\right) \right) =\pi _{j}\left( x_{1},\ldots ,x_{I}\right) $ for all $%
m_{0j}\in M_{0j}$ and all $\left( x_{1},\ldots ,x_{I}\right) \in X$. Because 
$\pi ^{\ast }=\left( \pi _{1}^{\ast },\ldots \pi _{J}^{\ast }\right) \in \Pi
_{B}^{\ast }$, it is agent $i$'s continuation-equilibrium strategy to use 
\begin{equation}
c_{ij}^{\ast }(\gamma _{j}^{\ast },\gamma _{-j},x_{i})=\phi _{ij}(x_{i})
\label{thm_robust_0}
\end{equation}%
for communicating with principal $j$ for all $\gamma _{-j}\in \Gamma _{-j}$
and all $x_{i}\in X_{i}.$ Since principapl $j$ has no role in determining
his action in $\gamma _{j}$, I let $c_{0j}^{\ast }(\gamma _{j}^{\ast
})=m_{0j}$ for some $m_{0j}\in M_{ij}$. This implies that $\pi ^{\ast
}=\left( \pi _{1}^{\ast },\ldots \pi _{J}^{\ast }\right) \in \Pi ^{B}$ is
principal $j$'s utility on the equilibrium path in a competing general
competing mechanism game with $\Gamma $.

Suppose that principal $j$ deviates to $\gamma _{j}^{\prime }\in \Gamma _{j}$%
. Conditional on a profile of strategies of communicating with non-deviating
principals $c_{-j}^{\ast }(\gamma _{-j}^{\ast },\gamma _{j}^{\prime })$,
suppose that agents choose a profile of strategies of communcating with
principal $j$, $c_{j}(\gamma _{j}^{\prime },\gamma _{-j}^{\ast })\in 
\mathcal{C}_{j}(\gamma _{j}^{\prime },\gamma _{-j}^{\ast }|c_{-j}^{\ast
}(\gamma _{j}^{\prime },\gamma _{-j}^{\ast }))$. The strategy profile $%
c_{j}(\gamma _{j}^{\prime },\gamma _{-j}^{\ast })$ induces a BIC direct
mechanism $g_{j}\left( \gamma _{j}^{\prime },c_{j}(\gamma _{j}^{\prime
},\gamma _{-j}^{\ast })\right) \in \Pi _{j}^{B}$. Because of (\ref%
{thm_robust_0}), the continuation-equilibrium utility for principal $j$
becomes 
\begin{equation*}
\mathbb{E}_{x}\left[ v_{j}\left( g_{j}\left( \gamma _{j}^{\prime
},c_{j}(\gamma _{j}^{\prime },\gamma _{-j}^{\ast },x)\right) ,\pi
_{-j}^{\ast }(x),x\right) \right] .
\end{equation*}%
Because $\pi ^{\ast }=\left( \pi _{1}^{\ast },\ldots \pi _{J}^{\ast }\right)
\in \Pi _{B}^{\ast }$, I have that 
\begin{multline}
\mathbb{E}_{x}\left[ v_{j}(\pi _{j}^{\ast }(x),\pi _{-j}^{\ast }(x),x)\right]
\geq  \label{thm_robust_1} \\
\mathbb{E}_{x}\left[ v_{j}\left( g_{j}\left( \gamma _{j}^{\prime
},c_{j}(\gamma _{j}^{\prime },\gamma _{-j}^{\ast },x)\right) ,\pi
_{-j}^{\ast }(x),x\right) \right] \text{ }\forall c_{j}(\gamma _{j}^{\prime
},\gamma _{-j}^{\ast })\in \mathcal{C}_{j}(\gamma _{j}^{\prime },\gamma
_{-j}^{\ast }|c_{-j}^{\ast }(\gamma _{j}^{\prime },\gamma _{-j}^{\ast })).
\end{multline}%
The inequality above show that $\pi ^{\ast }=\left( \pi _{1}^{\ast },\ldots
\pi _{J}^{\ast }\right) \in \Pi _{MR}^{\ast }$ because the left-hand side is
principal $j$'s utility on the equilibrium path in the competing general
competing mechanism game with $\Gamma $ given (\ref{thm_robust_0}).

Second, I prove $\Pi _{MR}^{\ast }\subseteq \Pi _{B}^{\ast }.$ Fix $\pi
^{\ast }=\left( \pi _{1}^{\ast },\ldots \pi _{J}^{\ast }\right) \in \Pi
_{MR}^{\ast }$. Then, (\ref{thm_robust_1}) is satisfied. For all $\gamma
_{j}\in \Gamma _{j}$, let 
\begin{equation*}
G_{j}(\gamma _{j})=\left\{ g_{j}(\gamma _{j},c_{j}(\gamma _{j},\gamma
_{-j}^{\ast }))\in \Pi _{j}^{B}:c_{j}(\gamma _{j},\gamma _{-j}^{\ast })\in 
\mathcal{C}_{j}(\gamma _{j},\gamma _{-j}^{\ast })\right\} .
\end{equation*}%
be the set of BIC direct mechanisms that are induced by all possible
continuation equilibrium communication strategy profiles. Because of Lemma %
\ref{lemma_eq_BIC}, $\tbigcup\limits_{\gamma _{j}\in \Gamma _{j}}G_{j}\left(
\gamma _{j}\right) \subseteq \Pi _{j}^{B}.$ To see $\tbigcup\limits_{\gamma
_{j}\in \Gamma _{j}}G_{j}\left( \gamma _{j}\right) \supseteq \Pi _{j}^{B},$
pick any $\pi _{j}\in \Pi _{j}^{B}$. Given a homeomorphism $\phi
_{ij}:X_{i}\rightarrow M_{ij},$ I choose a general mechanism $\gamma
_{j}\left( m_{0j},\phi _{1j}\left( x_{1}\right) ,\ldots ,\phi _{Ij}\left(
x_{I}\right) \right) =\pi _{j}\left( x_{1},\ldots ,x_{I}\right) $ for all $%
m_{0j}\in M_{0j}$ and all $\left( x_{1},\ldots ,x_{I}\right) \in X$. Then, $%
c_{j}(\gamma _{j},\gamma _{-j}^{\ast })=\left\{ c_{j}(\gamma
_{j}),c_{1j}(\gamma _{j},\gamma _{-j}^{\ast }),\ldots ,c_{Ij}(\gamma
_{j},\gamma _{-j}^{\ast })\right\} =\{m_{0j},\phi _{1j},\ldots ,\phi _{Ij}\}$
is a profile of continuation-equilibrium strategies of communicating with
principal $j$ and $\pi _{j}=g_{j}(\gamma _{j},c_{j}(\gamma _{j},\gamma
_{-j}^{\ast }))\in \tbigcup\limits_{\gamma _{j}\in \Gamma _{j}}G_{j}\left(
\gamma _{j}\right) $. Therefore, I have 
\begin{equation}
\tbigcup\limits_{\gamma _{j}\in \Gamma _{j}}G_{j}\left( \gamma _{j}\right)
=\Pi _{j}^{B}  \label{thm_robust2}
\end{equation}%
(\ref{thm_robust_1}) and (\ref{thm_robust2}) imply that 
\begin{equation*}
\mathbb{E}_{x}\left[ v_{j}(\pi _{j}^{\ast }(x),\pi _{-j}^{\ast }(x),x)\right]
\geq \mathbb{E}_{x}\left[ v_{j}\left( \pi _{j}^{\prime }(x),\pi _{-j}^{\ast
}(x),x\right) \right] \text{ }\forall \pi _{j}^{\prime }\in \Pi _{j}^{B},
\end{equation*}%
which implies that $\pi ^{\ast }=\left( \pi _{1}^{\ast },\ldots \pi
_{J}^{\ast }\right) \in \Pi _{B}^{\ast }$.
\end{proof}

\bigskip

As long as our interest is the set of market-information invariant robust
PBE allocations, it is enough to focus on equilibrium allocations in
(truth-telling) PBEs of the competing BIC direct mechanism game. Such
equilibrium allocations are supported when the BIC direct mechanism each
non-deviating principal implements does not depend on whether a competing
principal deviates or not. Furthermore, equilibrium allocations in
(truth-telling) PBEs of the competing BIC direct mechanism game is free from
Szentes' critique since $\Pi _{MR}^{\ast }=\Pi _{B}^{\ast }$.

\subsection{Uniquely implementable robust PBE allocations}

Fix a uniquely implementable robust PBE $\left( \gamma ^{\ast },c^{\ast
}\right) $. $\pi _{j}^{\ast }=g_{j}(\gamma _{j}^{\ast },c_{j}^{\ast }(\gamma
^{\ast }))$ is principal $j$'s unique BIC direct mechanism on the
equilibrium path. Given the additive separability of each agent's utility
function, $\pi _{j}^{\ast }$ is also his unique BIC direct mechanism off the
path following a competing principal's deviation. Let $\Pi _{UI}^{\ast }$ be
the set of all market-information invariant robust PBE allocations in a
general competing mechanism game with $\Gamma $: 
\begin{equation*}
\Pi _{UR}^{\ast }:=\left\{ 
\begin{array}{c}
\left( g_{k}\left( \gamma _{k}^{\ast },c_{01}^{\ast }(\gamma _{k}^{\ast
}),c_{1k}^{\ast }\left( \gamma ^{\ast },\cdot \right) ,\ldots ,c_{Ik}^{\ast
}\left( \gamma ^{\ast },\cdot \right) \right) \right) _{k\in \mathcal{J}}\in
\Pi ^{B}: \\ 
(\gamma ^{\ast },c^{\ast })\text{ is a uniquely implementable robust PBE.}%
\end{array}%
\right\} \text{.}
\end{equation*}

Applying the results on Bayesian implementation in the literature, it is
easy to see that a uniquely implementable robust PBE allocation $\pi ^{\ast
}=\left( \pi _{1}^{\ast },\ldots \pi _{J}^{\ast }\right) $ in a competing
mechanism game with $\Gamma $ must satisfy

\begin{enumerate}
\item for all $j\in \mathcal{J}$, $\pi _{j}^{\ast }$ satisfies Bayesian
monotonicity (Definition 3.1 in Palfrey and Srivastava (2018))

\item $\pi ^{\ast }$ is a truth-telling PBE BIC direct mechanisms in the
competing mechanism game with $\Pi ^{B}$ (i.e., $\pi ^{\ast }\in \Pi
_{B}^{\ast }$)
\end{enumerate}

That is, 
\begin{equation}
\Pi _{UR}^{\ast }\subseteq \Pi _{UB}^{\ast },  \label{UR1}
\end{equation}%
where 
\begin{equation*}
\Pi _{UB}^{\ast }:=\left\{ 
\begin{array}{c}
\pi =\left( \pi _{1},\ldots ,\pi _{J}\right) \in \Pi _{B}^{\ast }: \\ 
\pi _{j}\text{ satisfies Bayesian monotonicity }\forall j\in \mathcal{J}%
\end{array}%
\right\}
\end{equation*}%
(\ref{UR1}) is an extension of the necessary conditions for Bayesian
implementation in Theorem 3.1 in Palfrey and Srivastava (2018). Furthermore,
when additional conditions on the information structure on the type space
are satisfied (See Theorem 3.2 in Palfrey and Srivastava (2018)), conditions
1 and 2 are sufficient conditions for a uniquely implementable robust PBE
allocation in a competing mechanism game with $\Gamma $. That is, 
\begin{equation}
\Pi _{UR}^{\ast }\supseteq \Pi _{UB}^{\ast }  \label{UR2}
\end{equation}%
All these results can be adopted seemlessly given the additive separability
of each agent's utility function. From (\ref{UR1}) and (\ref{UR2}), we have 
\begin{equation}
\Pi _{UR}^{\ast }=\Pi _{UB}^{\ast }  \label{UR3}
\end{equation}%
Note that 
\begin{equation}
\Pi _{UB}^{\ast }\subset \Pi _{B}^{\ast }.  \label{UR4}
\end{equation}

Given $\pi ^{\ast }=\left( \pi _{1}^{\ast },\ldots \pi _{J}^{\ast }\right) $
that satisfies conditions 1 and 2, each principal $j$'s equilibrium
(indirect) mechanism can be constructed the same way that an indirect
mechanism is constructed to implement a BIC direct mechanism with Bayesian
monotonicity in a model with a single principal. In such an equilibrium
mechanism, a principal does not send a message to himself and agents'
messages fully determines his action.

Furthermore, a principal only need to evaluate his continuation-equilibrium
payoff upon deviation to a BIC direct mechanism for a profitable deviation
even if a deviation to a general mechanism is possible. Therefore, a
uniquely implementable robust PBE is free from Szentes' critique.

\section{\label{sec: robust_PBE}Robust PBE allocations}

Not every robust PBE allocation is market-information invariant. Some of
them can be supported only when agents play a contiuation equilibrium upon a
principal's deviation that make non-deviating principals change their BIC
direct mechanisms. In this section, I fully characterize the set of robust
PBE allocations in the general competing mechanism game. The greatest lower
bound of each principals' robust PBE utility in the general competing
mechanism game is greater than that in the standard competing mechanism
game. The same thing can be shown for the notion of PBE. One advantage of
the notion of robust PBE is that the incentive compatibility alone is
sufficient to characterize the set of robust equilibrium allocations given
the additive separability of agents' utility functions. This point is
explained in details in Section \ref{sec:discussion}.

Let $\Pi _{R}^{\ast }$ be the set of all robust PBE allocations of a general
competing mechanism game with $\Gamma $:%
\begin{equation*}
\Pi _{R}^{\ast }:=\left\{ \left( g_{k}\left( \gamma _{k}^{\ast
},c_{01}^{\ast }(\gamma _{k}^{\ast }),c_{1k}^{\ast }\left( \gamma ^{\ast
},\cdot \right) ,\ldots ,c_{Ik}^{\ast }\left( \gamma ^{\ast },\cdot \right)
\right) \right) _{k\in \mathcal{J}}\in \Pi ^{B}:(\gamma ^{\ast },c^{\ast })%
\text{ is a robust PBE.}\right\} .
\end{equation*}%
Our main result is that $\Pi _{R}^{\ast }$ is characterized by only BIC
direct mechanisms without reference to $\Gamma $. For all $j\in \mathcal{J}$%
, let 
\begin{equation}
\underline{V}_{j}=\min_{\pi _{-j}\in \Pi _{-j}^{B}}\max_{\pi _{j}\in \Pi
_{j}^{B}}\mathbb{E}_{x}\left[ v_{j}(\pi _{j}(x),\pi _{-j}(x),x)\right] .
\label{lower_bound}
\end{equation}%
Because $\mathbb{E}_{x}\left[ v_{j}(\pi _{1}(x),\ldots ,\pi _{J}(x),x)\right]
$ is the expected utility and each $\Pi _{k}^{B}$ is compact (Proposition %
\ref{Prop_compact_BIC}), I can then apply Berge's Maximum Theorem
(Aliprantis and Border (1999)) to show that $\min_{\pi _{-j}\in \Pi
_{-j}^{B}}\max_{\pi _{j}\in \Pi _{j}^{B}}\mathbb{E}_{x}\left[ v_{j}(\pi
_{j}(x),\pi _{-j}(x),x)\right] $ is well defined.

Define 
\begin{equation*}
\Pi _{RB}^{\ast }:=\left\{ \pi \in \Pi ^{B}:\mathbb{E}_{x}\left[ v_{j}(\pi
(x),x)\right] \geq \underline{V}_{j}\text{ for all }j\in \mathcal{J}\right\}
\end{equation*}

\begin{theorem}
\label{theorem1} 
\begin{equation}
\Pi _{R}^{\ast }=\Pi _{RB}^{\ast }.  \label{PBEallocations1}
\end{equation}
\end{theorem}

\begin{proof}
I first show that any robust PBE allocation must generate a utility no less
than $\underline{V}_{j}$ for all $j\in \mathcal{J}$. Fix a robust PBE $%
(\gamma ^{\ast },c^{\ast })$. On contrary, suppose that for some $j\in 
\mathcal{J}$ 
\begin{equation}
\mathbb{E}_{x}\left[ v_{j}\left( g_{j}\left( \gamma _{j}^{\ast },c_{j}^{\ast
}(\gamma ^{\ast },x)\right) ,g_{-j}\left( \gamma _{-j}^{\ast },c_{-j}^{\ast
}(\gamma ^{\ast },x)\right) ,x\right) \right] <\underline{V}_{j}.
\label{inequality1}
\end{equation}%
Suppose that principal $j$ unilaterally deviates to $\gamma _{j}^{\prime
}:M_{j}\rightarrow A_{j}$ such that, with a slight abuse of notation, (i) $%
M_{0j}=\Pi _{j}^{B}$, (ii) $M_{ij}=X_{i}$ for all $i\in \mathcal{I}$, (iii)
for all $\pi _{j}\in M_{0j}$ and all $x_{ij}\in X_{i}$ for all $i\in 
\mathcal{I}$,%
\begin{equation}
\gamma _{j}^{\prime }\left( \pi _{j},x_{1j},\ldots ,x_{Ij}\right) =\pi
_{j}\left( x_{1j},\ldots ,x_{Ij}\right) .  \label{deviator_mechanism}
\end{equation}%
This is a general mechanism where principal $j$ chooses his direct mechanism
from $\Pi _{j}^{B}$ and agents send their type messages at the same time.%
\footnote{%
A precise formulation of such a mechanism $\gamma _{j}^{\prime
}:M_{j}\rightarrow A_{j}$ can be given as follows. $M_{ij}$ is homeomorphic
to $X_{i},$ there exists a homeomorphism $\xi _{ij}:X_{i}\rightarrow M_{ij}$
and $M_{0j}$ is homeomorphic to $\eta _{j}(\Pi _{j}^{B})$, that is, there
exists an homeomorphism $\varrho _{0j}:\eta _{j}(\Pi _{j}^{B})\rightarrow
M_{0j}$, where $\eta _{j}$ is an embedding $\eta _{j}:\Pi
_{j}^{B}\rightarrow \overline{\Gamma }_{j}$.}

Now I can construct continuation-equilibrium strategies of communicating
with principal $j$ that provides a higher utility for principal $j$ than his
PBE utility, \emph{given} continuation-equilibrium strategies $c_{-j}^{\ast
}(\gamma _{j}^{\prime },\gamma _{-j}^{\ast })$ of communicating with
non-deviating principals. For this, I first show that I can keep
continuation-equilibrium strategies $c_{-j}^{\ast }(\gamma _{j}^{\prime
},\gamma _{-j}^{\ast })$ of communicating with non-deviating principals
regardless of continuation-equilibrium strategies of communicating with
principal $j$.

Note that principal $j$'s deviation to $\gamma _{j}^{\prime }$ is not
observable by non-deviating principals so that they do not change their
messages to themselves, i.e., $c_{0k}(\gamma _{k}^{\ast })=c_{0k}^{\ast
}(\gamma _{k}^{\ast })$ for all $k\neq j.$ Now consider agents'
communication with every principal $k\neq j.$ Because $c^{\ast }(\gamma
_{j}^{\prime },\gamma _{-j}^{\ast })$ is a continuation equilibrium, I have
that for all $k\neq j,$ all $i\in \mathcal{I}$, a.e. $x_{i}\in X_{i}$ and
all $\mu _{ik}\in \Delta (M_{ik}),$ 
\begin{multline*}
\mathbb{E}_{x_{-i}}\left[ u_{ik}\left( g_{k}\left( \gamma _{k}^{\ast
},c_{0k}^{\ast }(\gamma _{k}^{\ast }),c_{1k}^{\ast }(\gamma _{j}^{\prime
},\gamma _{-j}^{\ast },x_{1}),\ldots ,c_{Ik}^{\ast }(\gamma _{j}^{\prime
},\gamma _{-j}^{\ast },x_{I})\right) ,x_{-i},x_{i}\right) |x_{i}\right] \geq
\\
\mathbb{E}_{x_{-i}}\left[ u_{ik}\left( g_{k}\left( \gamma _{k}^{\ast
},c_{0k}^{\ast }(\gamma _{k}^{\ast }),\mu _{ik},c_{-ik}^{\ast }(\gamma
_{j}^{\prime },\gamma _{-j}^{\ast },x_{-i})\right) ,x_{-i},x_{i}\right)
|x_{i}\right] .
\end{multline*}%
This implies that I can also keep each agent $i$'s strategy $c_{ik}^{\ast
}(\gamma _{j}^{\prime },\gamma _{-j}^{\ast },\cdot )$ of communicating with
each non-deviating principal $k$ for continuation-equilibrium strategies of
communicating with principal $j$ at $(\gamma _{j}^{\prime },\gamma
_{-j}^{\ast })$. Therefore, I fix 
\begin{equation}
c_{-j}^{\ast }(\gamma _{j}^{\prime },\gamma _{-j}^{\ast }).
\label{alternative_str1}
\end{equation}%
for continuation-equilibrium strategies of communicating with non-deviating
principals regardless of continuation-equilibrium strategies of
communicating with principal $j$.

Now consider how to construct continuation-equilibrium strategies of
communicating with principal $j$ that are best for principal $j$ upon his
deviation to $\gamma _{j}^{\prime }$. Note that $g_{k}(\gamma _{k}^{\ast
},c_{k}^{\ast }(\gamma _{j}^{\prime },\gamma _{-j}^{\ast }))$ is
individually BIC for all $k\neq j$ by Lemma \ref{lemma_eq_BIC}. Then, with a
slight abuse of notation, let principal $j$ choose a degenerate probability
distribution $c_{0j}(\gamma _{j}^{\prime })\in M_{0j}=\Pi _{j}^{B}$
satisfying 
\begin{equation}
c_{0j}(\gamma _{j}^{\prime })\in \underset{\pi _{j}\in \Pi _{j}^{B}}{\arg
\max }\;\mathbb{E}_{x}\left[ v_{j}\left( \pi _{j}(x),g_{-j}\left( \gamma
_{-j}^{\ast },c_{-j}^{\ast }(\gamma _{j}^{\prime },\gamma _{-j}^{\ast
},x)\right) ,x\right) \right] ,  \label{alternative_str2}
\end{equation}%
and let each agent choose a non-degenerate probability distribution that
chooses her true type with probability one and send it to principal $j$,
i.e., for all $i\in \mathcal{I}$ and a.e. $x_{i}\in X_{i}$%
\begin{equation}
c_{ij}(\gamma _{j}^{\prime },\gamma _{-j}^{\ast },x_{i})=x_{i}\text{.}
\label{alternative_str3}
\end{equation}%
Principal $j$ takes $g_{k}\left( \gamma _{k}^{\ast },c_{k}^{\ast }(\gamma
_{j}^{\prime },\gamma _{-j}^{\ast },\cdot )\right) $ for all $k\neq j$ as
given. Given agents' truthful type reporting to him, $c_{0j}(\gamma
_{j}^{\prime })$ satisfying (\ref{alternative_str2}) is clearly principal $j$%
's optimal choice of a direct mechanism in $M_{0j}=\Pi _{j}^{B}$. Given
principal $j$'s choice of a BIC direct mechanism $\pi _{j}=c_{0j}(\gamma
_{j}^{\prime })$, truthful type reporting (i.e., (\ref{alternative_str3}))
is also optimal for every agent $i$ of a.e. $x_{i}\in X_{i}$. Therefore, (%
\ref{alternative_str1}), (\ref{alternative_str2}) and (\ref{alternative_str3}%
) constitute a \emph{continuation equilibrium} at $(\gamma _{j}^{\prime
},\gamma _{-j}^{\ast })$.

Because $g_{-j}\left( \gamma _{-j}^{\ast },c_{-j}^{\ast }(\gamma
_{j}^{\prime },\gamma _{-j}^{\ast })\right) \in \Pi _{-j}^{B}$, principal $j$%
's utility in this continuation equilibrium upon deviation to $\gamma
_{j}^{\prime }$ is no less than $\underline{V}_{j}$: 
\begin{equation}
\underset{\pi _{j}\in \Pi _{j}^{B}}{\max }\mathbb{E}_{x}\left[ v_{j}\left(
\pi _{j}(x),g_{-j}\left( \gamma _{-j}^{\ast },c_{-j}^{\ast }(\gamma
_{j}^{\prime },\gamma _{-j}^{\ast },x)\right) ,x\right) \right] \geq 
\underline{V}_{j}  \label{inequality2}
\end{equation}%
(\ref{inequality1}) and (\ref{inequality2}) show that 
\begin{multline*}
\underset{\pi _{j}\in \Pi _{j}^{B}}{\max }\mathbb{E}_{x}\left[ v_{j}\left(
\pi _{j}(x),g_{-j}\left( \gamma _{-j}^{\ast },c_{-j}^{\ast }(\gamma
_{j}^{\prime },\gamma _{-j}^{\ast },x)\right) ,x\right) \right] \mathbb{>} \\
\mathbb{E}_{x}\left[ v_{j}\left( g_{j}\left( \gamma _{j}^{\ast },c_{j}^{\ast
}(\gamma ^{\ast },x)\right) ,g_{-j}\left( \gamma _{-j}^{\ast },c_{-j}^{\ast
}(\gamma ^{\ast },x)\right) ,x\right) \right]
\end{multline*}%
and hence principal $j$'s utility in this continuation equilibrium is
greater than his equilibrium utility. Because this continuation equilibrium
upon deviation to $\gamma _{j}^{\prime }$ is constructed with (\ref%
{alternative_str1}), (\ref{robust_eq_cond}) in the definition of a robust
PBE is violated. Therefore, any robust PBE allocation must be in the set on
the right hand side of (\ref{PBEallocations1})

To complete the proof, I need to show how to construct a robust PBE to
support an allocation in the set on the right hand side of (\ref%
{PBEallocations1}). Consider an allocation $\pi ^{\ast }=(\pi _{1}^{\ast
},\ldots ,\pi _{J}^{\ast })$ in the set on the right hand side of (\ref%
{PBEallocations1}). For each $k\in \mathcal{J}$, I consider the following 
\emph{deviator-reporting mechanism} $\gamma _{k}^{\bigstar }:M_{1k}\times
\cdots \times M_{Ik}\rightarrow \Delta (A_{k})$, where, with a slight abuse
of notation, $M_{ik}=\{1,\ldots ,J\}\times X_{i}$ (precisely, $M_{ik}$ is
homeomorphic to $\{1,\ldots ,J\}\times X_{i},$ that is there exists a
homeomorphism $\vartheta _{ik}:\{1,\ldots ,J\}\times X_{i}\rightarrow M_{ik}$%
). Note that a deviator-reporting mechanism is a standard mechanism in $%
\overline{\Gamma }_{k}$ where the principal does not send a message to
himself. In this deviator-reporting mechanism, each agent sends a message on
who the deviator is, if any, and her type. If more than a half of agents
report $j\neq k$ as the deviator, principal $k$ assigns a BIC direct
mechanism $\pi _{k}^{j}$ that punishes principal $k$. Otherwise, he assigns $%
\pi _{k}^{\ast }.$ I choose the following BIC direct mechanism for $\pi
_{k}^{j}$. For all $j\in \mathcal{J}$, then $\pi _{k}^{j}$ is principal $k$%
's BIC direct mechanism in $\pi _{-j}^{j}$ satisfying%
\begin{equation}
\pi _{-j}^{j}=(\pi _{k}^{j})_{k\neq j}\in \underset{\pi _{-j}\in \Pi
_{-j}^{B}}{\arg \min }\left[ \max_{\pi _{j}\in \Pi _{j}^{B}}\mathbb{E}_{x}%
\left[ v_{j}(\pi _{j}(x),\pi _{-j}(x),x)\right] \right] .  \label{arg_min}
\end{equation}%
(\ref{arg_min}) implies that 
\begin{equation}
\mathbb{E}_{x}\left[ v_{j}\left( \pi ^{\ast }(x),x\right) \right] \geq
\max_{\pi _{j}\in \Pi _{j}^{B}}\mathbb{E}_{x}\left[ v_{j}\left( \pi
_{j}(x),\pi _{-j}^{j}(x),x\right) \right] =\underline{V}_{j}
\label{eq_payoff}
\end{equation}%
because $\pi ^{\ast }=(\pi _{1}^{\ast },\ldots ,\pi _{J}^{\ast })$ in the
set on the right hand side of (\ref{PBEallocations1}).

Principal $k$'s deviator-reporting mechanism is constructed as follows. For
all $((\ell _{1k},$ $x_{1k}),\ldots ,(\ell _{Ik},x_{Ik}))\in \times _{i\in 
\mathcal{I}}M_{ik},$ 
\begin{equation}
\gamma _{k}^{\bigstar }((\ell _{1k},x_{1k}),\ldots ,(\ell
_{Ik},x_{Ik}))=\left\{ 
\begin{array}{cc}
\pi _{k}^{j}(x_{1k},\ldots ,x_{Ik}) & \text{if }\#\{\ell _{ik}=j\text{ for }%
j\neq k\}>I/2 \\ 
\pi _{k}^{\ast }(x_{1k},\ldots ,x_{Ik}) & \text{otherwise}%
\end{array}%
\right. .  \label{deviator_reporting_DM}
\end{equation}%
Every principal $k$ offers the mechanism $\gamma _{k}^{\bigstar }$ specified
in (\ref{deviator_reporting_DM}). Suppose that no principal deviates. Then,
each agent $i$ of a.e. $x_{i}\in X_{i}$ truthfully reports $x_{i}$ to every
principal $k$ along with the message $\ell _{ik}=k.$ Then $\gamma _{k}^{\ast
}$ implements $\pi _{k}^{\ast }(x_{1},\ldots ,x_{I}).$ Since there are three
or more agents, any single agent $i$'s deviation from $\ell _{ik}=k$ does
not prevent principal $k$ from implementing $\pi _{k}^{\ast }(x_{1},\ldots
,x_{I})$. Because $\pi _{k}^{\ast }$ is individually BIC, such communication
behavior constitutes a continuation equilibrium and each principal $k$
receives the utility of $\mathbb{E}_{x}\left[ V_{k}\left( \pi ^{\ast
}(x),x\right) \right] $.

Suppose that principal $j$ deviates to a general mechanism $\gamma
_{j}:M_{j}\rightarrow A_{j}.$ Then, each agent $i$ of a.e. $x_{i}\in X_{i}$
truthfully reports $x_{i}$ to every principal $k\neq j$ along with the
message $\ell _{ik}=j$. Then, $\gamma _{k}^{\bigstar }$ implements $\pi
_{k}^{j}(x_{1},\ldots ,x_{I})$. Because $\pi _{k}^{j}$ is individually BIC,
I can fix such communication behavior for a continuation equilibrium upon
principal $j$'s deviation to $\gamma _{j}.$ Let $\left( c_{0j}(\gamma
_{j}),c_{1j}(\gamma _{j},\gamma _{-j}^{\bigstar },\cdot ),\ldots
,c_{Ij}(\gamma _{j},\gamma _{-j}^{\bigstar },\cdot )\right) $ be arbitrary
continuation-equilibrium strategies of communicating with principal $j$ upon 
$j$'s deviation to $\gamma _{j}$. Then, it is clear that 
\begin{equation}
g_{j}\left( \gamma _{j},c_{0j}(\gamma _{j}),c_{1j}(\gamma _{j},\gamma
_{-j}^{\bigstar },\cdot ),\ldots ,c_{Ij}(\gamma _{j},\gamma _{-j}^{\bigstar
},\cdot )\right) \in \Pi _{j}^{B}.  \label{BIC_DM}
\end{equation}%
Because of (\ref{BIC_DM}), I have that 
\begin{equation}
\max_{\pi _{j}\in \Pi _{j}^{B}}\mathbb{E}_{x}\left[ v_{j}\left( \pi
_{j}(x),\pi _{-j}^{j}(x),x\right) \right] \geq \mathbb{E}_{x}\left[
v_{j}\left( g_{j}\left( \gamma _{j},c_{j}(\gamma _{j},\gamma _{-j}^{\bigstar
},x)\right) ,\pi _{-j}^{j}(x),x\right) \right]  \label{deviation_payoff}
\end{equation}%
where $c_{j}(\gamma _{j},\gamma _{-j}^{\bigstar },x)=(c_{0j}(\gamma
_{j}),c_{1j}(\gamma _{j},\gamma _{-j}^{\bigstar },x_{1}),\ldots
,c_{Ij}(\gamma _{j},\gamma _{-j}^{\bigstar },x_{I})).$ Note that the
expression on the right hand side of (\ref{deviation_payoff}) is principal $%
j $'s utility in a continuation equilibrium upon deviation to $\gamma _{j}.$
Combining (\ref{eq_payoff}) and (\ref{deviation_payoff}) yields 
\begin{equation}
\mathbb{E}_{x}\left[ v_{j}\left( \pi ^{\ast }(x),x\right) \right] \geq 
\mathbb{E}_{x}\left[ v_{j}\left( g_{j}\left( \gamma _{j},c_{j}(\gamma
_{j},\gamma _{-j}^{\bigstar },x)\right) ,\pi _{-j}^{j}(x),x\right) \right] ,
\label{robust_PBE_payoff}
\end{equation}%
which shows that principal $j$ cannot gain upon deviation to $\gamma _{j}$. (%
\ref{robust_PBE_payoff}) implies that $\pi ^{\ast }=(\pi _{1}^{\ast },\ldots
,\pi _{J}^{\ast })$ is supported as a robust PBE allocation because the
first argument $g_{j}\left( \gamma _{j},c_{j}(\gamma _{j},\gamma
_{-j}^{\bigstar },x)\right) $ on the right-hand side of (\ref%
{robust_PBE_payoff}) is derived with any continuation-equilibrium strategies
of communicating with principal $j$, $\left( c_{0j}(\gamma
_{j}),c_{1j}(\gamma _{j},\gamma _{-j}^{\bigstar },\cdot ),\ldots
,c_{Ij}(\gamma _{j},\gamma _{-j}^{\bigstar },\cdot )\right) $, given agents'
truth telling to non-deviating principals.
\end{proof}

\bigskip

The proof of Theorem \ref{theorem1} first establishes that a principal's
utility in a robust PBE $(\gamma ^{\ast },c^{\ast })$ cannot be lower than
his minmax value over BIC direct mechanisms. Given continuation-equilibrium
strategies $c_{-j}^{\ast }$ of communicating with non-deviating principals,
principal $j$ can deviate to a general mechanism $\gamma _{j}^{\prime }$
where he can choose any BIC direct mechanism from $\Pi _{j}^{B}$ and agents
send their type messages to principal $j$ at the same time. Given the
robustness of PBE, principal $j$ can implement the BIC direct mechanism from 
$\Pi _{j}^{B}$ that maximizes his utility, taking $g_{-j}\left( \gamma
_{-j}^{\ast },c_{-j}^{\ast }(\gamma _{j}^{\prime },\gamma _{-j}^{\ast
})\right) \in \Pi _{-j}^{B}$ as given. This directly implies that principal $%
j$'s robust PBE utility cannot be lower than his minmax value over BIC
direct mechanisms.

The proof of Theorem \ref{theorem1} is completed by constructing the
deviator-reporting mechanism to show how to support any profile of BIC
direct mechanisms that provides every principal a utility no less than his
minmax value over BIC direct mechanisms as a robust PBE allocation.

The existence of a robust PBE with deviator-reporting mechanisms can be
addressed. First consider a continuation equilibrium upon principal $j$'s
deviation to a general mechanism $\gamma _{j}$. If principal $j$ deviates,
the truth-telling continuation equilibrium that implements $\pi _{-j}^{j}$
is also well established and exists when it comes to communicating with
non-deviating principals.

One may wonder if there is a (mixed-strategy) continuation equilibrium 
\begin{equation*}
\left( c_{0j}(\gamma _{j}),c_{1j}(\gamma _{j},\gamma _{-j}^{\bigstar },\cdot
),\ldots ,c_{Ij}(\gamma _{j},\gamma _{-j}^{\bigstar },\cdot )\right)
\end{equation*}%
exists for communicating with deviating principal $j$ upon his deviation to
a general mechanism $\gamma _{j}$. When agents and principal $j$ send
messages to principal $j$ given his mechanism $\gamma _{j}$, they take $\pi
_{-j}^{j}$ as given. For all $(m_{0j},m_{1j},\ldots ,m_{Ij})\in M_{j}$ and
all $x\in X$, principal $j$'s utility for the subgame defined by $\gamma
_{j} $ follows 
\begin{equation}
\mathcal{V}_{j}(m_{0j},m_{1j},\ldots ,m_{Ij},x):=v_{j}\left( \gamma
_{j}(m_{0j},m_{1j},\ldots ,m_{Ij}),\pi _{-j}^{j}(x),x\right) .
\label{prinicpal_payoff2}
\end{equation}%
Because the agent's utility function is additively separable with respect to
principals' actions, the agent's relevant utility for the subgame defined by 
$\gamma _{j}$ is 
\begin{equation}
\mathcal{U}_{ij}(m_{0j},m_{1j},\ldots ,m_{Ij},x):=u_{ij}(\gamma
_{j}(m_{0j},m_{1j},\ldots ,m_{Ij}),x).  \label{agent_payoff2}
\end{equation}%
If $\mathcal{V}_{j}$ and $\mathcal{U}_{ij}$ and the information structure
satisfy the sufficient conditions for the existence of a Bayesian
equilibrium with continuous actions identified in the literature (Balder
(1988), Hellman and Levy (2017), Milgrom and Weber (1985), etc.), then a
continuation equilibrium for the subgame defined by $\gamma _{j}$ exists.%
\footnote{%
For example, if $\mathcal{V}_{j}$ and $\mathcal{U}_{ij}$ are uniformly
continuous and the information structure is absolutely continuous (Milgrom
and Weber (1985)), then a (mixed-strategy) continuation equilibrium exists.
However, the existence of a (mixed-strategy) continuation equilibrium for
communicating with principal $j$ upon his deviation should come naturally.
When principal $j$ deviates, he will certainly deviate to a mechanism that
guarantees the existence of a (mixed-strategy) continuation equilibrium for
communicating with him.}

Because any (mixed) continuation-equilibrium strategies%
\begin{equation*}
\left( c_{0j}(\gamma _{j}),c_{1j}(\gamma _{j},\gamma _{-j}^{\bigstar },\cdot
),\ldots ,c_{Ij}(\gamma _{j},\gamma _{-j}^{\bigstar },\cdot )\right)
\end{equation*}%
upon deviation to any mechanism $\gamma _{j}$ induces a BIC direct mechanism 
\begin{equation}
g_{j}\left( \gamma _{j},c_{0j}(\gamma _{j}),c_{1j}(\gamma _{j},\gamma
_{-j}^{\bigstar },\cdot ),\ldots ,c_{Ij}(\gamma _{j},\gamma _{-j}^{\bigstar
},\cdot )\right) \in \Pi _{j}^{B},  \label{deviation_BIC}
\end{equation}%
the notion of robustness makes it possible for principal $j$ to focus on any
BIC direct mechanism in $\Pi _{j}^{B}$ for his deviation.

Finally, note that a deviator-reporting mechanism $\gamma _{k}^{\bigstar }$
can be viewed as one of mechanisms in $\Gamma _{k}$ for all $k$ and each
principal $j$ receives a utility no higher than $\underline{V}_{j}$ upon
deviation to any $\gamma _{j}\in \Gamma _{j}$ in a continuation equilibrium
where agents report the identity of a deviator and their types truthfully to
non-deviating principals whose mechanisms are deviator-reporting mechanisms.
Given the structure of a deviator-reporting mechanism and (\ref%
{deviation_BIC}), this implies that a robust PBE of a general competing
mechanism game with $\Gamma $ exists if (i) $\Pi _{RB}^{\ast }=\left\{ \pi
\in \Pi ^{B}:\mathbb{E}_{x}\left[ v_{j}(\pi (x),x)\right] \geq \underline{V}%
_{j}\text{ for all }j\in \mathcal{J}\right\} $ is non-empty\footnote{$\Pi
_{R}^{\ast }$ is non-empty because $\Pi _{B}^{\ast }\subset \Pi _{R}^{\ast }$
and $\Pi _{B}^{\ast }$ is non-empty because of Corollary \ref%
{corollary_existence}.} and (ii) $\mathcal{V}_{j}$ and $\mathcal{U}_{ij}$
defined for any given $\gamma _{j}$ ((\ref{prinicpal_payoff2}) and (\ref%
{agent_payoff2})), and the information structure satisfy the sufficient
conditions for the existence of a Bayesian equilibrium with continuous
actions identified in the literature (Balder (1988), Hellman and Levy
(2017), Milgrom and Weber (1985), etc.).

Note that it is clear that $\Pi _{B}^{\ast }\subset \Pi _{RB}^{\ast }$.
Applying (\ref{UR4}), we have that 
\begin{equation*}
\Pi _{UB}^{\ast }\subset \Pi _{B}^{\ast }\subset \Pi _{RB}^{\ast }.
\end{equation*}%
All three sets are independent of the set of general mechanisms allowed in
competing mechanism games and they are all specified in terms of BIC direct
mechanisms. My paper identifies the notions of PBEs that produce those sets
of allocations as the sets of equilibrium allocations in general competing
mechnaism games respectively.

\section{Discussion\label{sec:discussion}}

This paper proposes a general competing mechanism game of incomplete
information where a mechanism allows its designer to send a message to
himself at the same time agents send messages. This paper introduces various
notions of robust perfect Bayesian equilibrium (PBE). Given the additive
separability of each agent's utility function with respect to principals'
actions, this paper provides the full characterization of equilibrium
allocations in terms of Bayesian incentive compatible (BIC) direct
mechanisms for each notion of robust PBE, without reference to the set of
arbitrary mechanisms allowed in the game.

\subsection{Robust PBE in standard competing mechanism games}

Consider the standard competing mechanism games where each principal $j$
offer a standard mechanism in $\overline{\Gamma }_{j}$ where agents'
messages to $j$ fully determines $j$'s action. In the standard competing
mechanism game, the adoption of robust PBE as the solution concept does not
lead to the full characterization of equilibrium allocation in terms of BIC
direct mechanisms even with the separable utility functions for the agents.

To see this point, fix a robust PBE $(\overline{\gamma },\overline{c})$ in
the standard competing mechanism game. Suppose that principal $j$ deviates
to a standard mechanism $\gamma _{j}^{\prime }\in \overline{\Gamma }_{j}.$
Agents' continuation equilibrium strategies of communicating with principal $%
j$, $\bar{c}_{j}(\gamma _{j}^{\prime },\bar{\gamma}_{-j})$ induce a BIC
direct mechanism $g_{j}\left( \gamma _{j}^{\prime },\bar{c}_{j}(\gamma
_{j}^{\prime },\bar{\gamma}_{-j})\right) \in \Pi _{j}^{B}$, whereas agents'
continuation equilibrium strategies of communicating with principals except
for $j$, $\bar{c}_{j}(\gamma _{j}^{\prime },\bar{\gamma}_{-j})$ induce a
profile of BIC direct mechanisms $g_{-j}\left( \bar{\gamma}_{-j},\bar{c}%
_{j}(\gamma _{j}^{\prime },\bar{\gamma}_{-j})\right) \in \Pi _{-j}^{B}.$
Without knowing the structure of $\gamma _{j}^{\prime }$, we do not know the
set of all BIC direct mechanisms that can be induced by all possible
continuation equilibrium strategies of communicating with principal $j$, $%
c_{j}(\gamma _{j}^{\prime },\bar{\gamma}_{-j})$ even if agents' utility
functions are separable. This makes it not possible to provide the full
characterization of robust PBE allocations in terms of BIC direct mechanisms
even with the separable utility functions for agents.

Such a diffiulty does not prevent the full characterization of robust PBE
allocations in general competing mechanism games. Principal $j$ can alway
offer a general mechanism in which the message space for himself is the set
of all possible BIC direct mechanisms $\Pi _{j}^{B}$ and the message space
for each agent is her type space. Given the separable utility functions for
agents, the set of all possible BIC direct mechanisms for $j$ induced by all
continuation-equilibrium strategies of communicating with $j$ is then equal
to $\Pi _{j}^{B}$. This makes that principal $j$'s robust PBE utility cannot
be lower than 
\begin{equation*}
\underline{V}_{j}=\min_{\pi _{-j}\in \Pi _{-j}^{B}}\max_{\pi _{j}\in \Pi
_{j}^{B}}\mathbb{E}_{x}\left[ v_{j}(\pi _{j}(x),\pi _{-j}(x),x)\right] .
\end{equation*}%
I can then show that in the general competing mechanism game, any profile of
BIC direct mechanism that provides each principal $j$ with a utility no less
than $\underline{V}_{j}$ can be supported in a robust PBE where each
principal $j$ offers a deviator-reporting mechanism in which agents are
asked to report the identity of a deviating principal, if any, along with
her type.

\subsection{Mixed strategies for principals}

I have not considered mixed strategies for principals. This can be done
easily. Consider the general competing mechanism game. When principal $j$
deviates given that the other principals employ mixed strategies for their
mechanism choice, he can deviate to a general mechanism where he asks agents
not only about their types but also the BIC direct mechanisms $\pi _{-j}$
induced from the other principals' mechanisms. Principal $j$'s message in
this general mechanism is a mapping from $\Pi _{-j}^{B}$ into $\Pi _{j}^{B},$
which specifies his choice of a BIC direct mechanism as a function of
agents' reports on the other principals' BIC direct mechanisms. In this way,
even in a mixed-strategy robust PBE, a deviating principal can always
implement a BIC direct mechanism that maximizes his utility conditional on
BIC direct mechanisms induced by the other principals' mechanisms.
Therefore, the greatest lower bound of principal $j$'s mixed-strategy robust
PBE utility is the same as the greatest lower bound of his pure-strategy
robust PBE utility, that is, his minmax value $\underline{V}_{j}$ over BIC
direct mechanisms. For all $j$, let 
\begin{equation*}
\Phi _{j}(\sigma _{1},\ldots ,\sigma _{J}):=\int_{\Pi _{1}^{B}}\cdots
\int_{\Pi _{J}^{B}}\mathbb{E}_{x}\left[ v_{j}(\pi _{1}(x),\ldots \pi
_{J}(x),x)\right] d\sigma _{1}\cdots d\sigma _{J},
\end{equation*}%
where $(\sigma _{1},\ldots ,\sigma _{J})\in \Delta (\Pi _{1}^{B})\times
\cdots \times \Delta (\Pi _{J}^{B})$. Then, the set of mixed-strategy robust
PBE allocations of a general competing mechanism game with $\Gamma $ is%
\begin{equation*}
\left\{ (\sigma _{1},\ldots ,\sigma _{J})\in \Delta (\Pi _{1}^{B})\times
\cdots \times \Delta (\Pi _{J}^{B}):\Phi _{j}(\sigma _{1},\ldots ,\sigma
_{J})\geq \underline{V}_{j}\text{ for all }j\in \mathcal{J}\right\} .
\end{equation*}

\end{document}